\documentclass[pra,twocolumn,showpacs,amsfonts,amsmath,floatfix,superscriptaddress]{revtex4-2}

\usepackage[svgnames]{xcolor}

 \usepackage[
            colorlinks=true,        
            allcolors = black,  
            citecolor=DarkBlue, 
            linkcolor=DarkBlue,
            urlcolor=DarkBlue
        ]{hyperref}  
\usepackage{amsmath, amsfonts, amssymb, amsthm, bbm}
\usepackage{dsfont}
\usepackage{bm} 
\usepackage{mathtools}
\usepackage{times}
\usepackage{physics}
\usepackage{graphicx}
\usepackage{float}
\usepackage[colorinlistoftodos]{todonotes}
\usepackage{dcolumn}

\usepackage{diagbox}

\newcommand{\mG}{\mathcal{G}}

\newcommand{\bal}{\begin{equation}\begin{aligned}}
\newcommand{\eal}{\end{aligned}\end{equation}}

\newcommand{\GKP}{{\rm GKP}}
\newcommand{\STAB}{{\rm STAB}}
\newcommand{\DV}{{\rm DV}}
\newcommand{\CV}{{\rm CV}}
\newcommand{\cell}{{\rm cell}}

\theoremstyle{plain}
\newtheorem{theorem}{Theorem}

\newtheorem{proposition}{Proposition}

\theoremstyle{definition}

\allowdisplaybreaks

\begin{document}

\title{
Bridging magic and non-Gaussian resources via Gottesman-Kitaev-Preskill encoding}

\author{Oliver Hahn}%
\email{hahn@g.ecc.u-tokyo.ac.jp}
\affiliation{Wallenberg Centre for Quantum Technology, Department of Microtechnology and Nanoscience, Chalmers University of Technology, Sweden , SE-412 96 G\"{o}teborg, Sweden}
\affiliation{Department of Basic Science, The University of Tokyo, 3-8-1 Komaba, Meguro-ku, Tokyo, 153-8902, Japan}

\author{Giulia Ferrini}
\affiliation{Wallenberg Centre for Quantum Technology, Department of Microtechnology and Nanoscience, Chalmers University of Technology, Sweden , SE-412 96 G\"{o}teborg, Sweden}

\author{Ryuji Takagi}%
\email{ryujitakagi.pat@gmail.com}
\affiliation{Department of Basic Science, The University of Tokyo, 3-8-1 Komaba, Meguro-ku, Tokyo, 153-8902, Japan}

\begin{abstract}
Although the similarity between non-stabilizer states---also known as magic states---in discrete-variable systems and non-Gaussian states in continuous-variable systems has widely been recognized, the precise connections between these two notions have still been unclear. We establish a fundamental link between these two quantum resources via the Gottesman-Kitaev-Preskill (GKP) encoding. We show that the negativity of the continuous-variable Wigner function for an encoded GKP state coincides with a magic measure we introduce, which matches the negativity of the discrete Wigner function for odd dimensions. We also provide a continuous-variable representation of the stabilizer R\'enyi entropy---a recent proposal for a magic measure for multi-qubit states. 
We further provide the magic measure with an operational interpretation as a runtime of a classical simulation algorithm. In addition, we employ our results to prove that implementing a multi-qubit logical non-Clifford operation in the GKP code subspace requires a non-Gaussian operation even at the limit of perfect encoding, despite the fact that the ideal GKP states already come with a large amount of non-Gaussianity.
\end{abstract}

\maketitle

\section{Introduction}
\label{sec:introduction}

The difference between what constitutes quantum and classical physics is often hard to grasp. 
The hope is to leverage quantum mechanics in order to get a computational speed-up when using quantum computing compared to classical computation.
Finding and pinpointing the origins of the speed-up or what property allows for such a phenomenon is still an open problem. Aside from an academic interest, this undertaking would allow us to identify and quantify what resources are required to do a certain computational task.
This fact becomes even more important, as in reality every quantum information processing task will be restricted in a certain way given that it will be implemented in a physical system, subjected to experimental constraints.

One of the promising paradigms for quantum information processing assumes infinite dimensional Hilbert spaces associated with observable possessing a continuous-variable spectrum. Such a model can  be implemented using for example quantum optical systems~\cite{konno2024logical}, as well as  microwave cavities coupled to superconducting systems~\cite{campagne2020quantum, sivak2023real} or trapped ions~\cite{Fl_hmann_2019, de2022error}, among others. In such systems, non-Gaussian components have been identified as necessary resources for quantum computational advantages, as computation solely run by Gaussian resources can be efficiently simulated classically~\cite{mari2012positive,Veitch_2013,PhysRevLett.88.097904, PhysRevX.6.021039}. 
Such non-Gaussian features in e.g. quantum states can be quantified by several measures of non-Gaussianity~\cite{Genoni2007measure,Genoni2008quantifying,takagi2018convex, albarelli2018resource,Chabaud2020stellar,Regula2021operational, Lami2021framework}, among which the negativity of Wigner function~\cite{Kenfack2004negativity,takagi2018convex, albarelli2018resource} has been known as a computable measure that is necessary for exponential quantum advantage~\cite{Veitch_2013}.  

The other paradigm for quantum information processing assumes discrete-variable systems, in which quantum information is encoded in finite-dimensional Hilbert spaces, and can be implemented in a vast amount of platforms~\cite{ladd2010quantum}. 
Among many relevant quantum resources needed for efficient quantum information processing in discrete-variable systems, one peculiar quantity necessary for quantum speedup is the non-stabilizerness~\cite{bravyi2005universal}, also known as quantum magic, which stems from the fact that quantum circuits only consisting of stabilizer states and Clifford operations can be efficiently simulable by classical computers~\cite{Gottesman1998Heisenberg}. 
Interestingly, for odd-dimensions the magicness of discrete-variable states can also be studied by looking at a discrete version of the Wigner function~\cite{gross2005finite,gross2006hudson} analogously to the case of continuous-variable systems.
Indeed, the negativity of the discrete Wigner function~\cite{Veitch2012negative,veitch2014resource} has been shown to be a valid magic measure when the underlying Hilbert space has odd dimensions.  
For even dimensions, one needs to consider other meausures~\cite{howard2017application,bravyi2016improved, bravyi2016trading,Regula2018convex,bravyi2019simulation,Bu2019efficient,beverland2020lower,seddon2021quantifying, leone2022stabilizer,hahn2022quantifying,Haug2023scalable,Bu2023stabilizer,Haug2023efficient,Bu2024entropic}, as there is no known quasiprobability distribution that easily connects to magic.

Although some conceptual similarities between non-Gaussianity and magic resources have been observed~\cite{gross2006hudson,Bu2023quantum,Bu2023discrete,Bu2024entropic}, the direct quantitative connection between these two resources has still been elusive. 
In particular, constructing a map between magic and non-Gaussianity would strengthen the relation between two main operational frameworks that are important for quantum computing and provide a novel approach where one resource could be analyzed by employing tools developed for analyzing the other.

In this work, we accomplish this mapping by finding a fundamental relation between the discrete and continuous-variable systems via the Gottesman-Kitaev-Preskill (GKP) encoding~\cite{gottesman2001encoding}, which is one of the most promising error-correcting codes for continuous-variable systems.
We introduce a family of distributions for discrete-variable systems and show that their $l_p$-norm exactly corresponds to that of the continuous-variable Wigner function for the GKP state encoding the original discrete-variable qudit. Specifically for odd dimension, the $l_1$ norm of the qudit distributions yields the negativity of the associated Wigner function for both discrete and continuous-variable settings. 
The connection is even stronger as the continuous Wigner function can be directly represented using the discrete Wigner function of the encoded state. 
On the other hand, our distributions yield a magic measure for all dimensions, which encompasses the negativity of the discrete Wigner function defined for odd dimensions and the stabilizer R\'enyi entropy~\cite{leone2022stabilizer} defined for multi-qubit systems in a unified manner.
These results allow us to provide a fundamental and quantitative relation between magic and non-Gaussianity.
Our results, therefore, allow for recovering and significantly extending a recent finding of the relation between multi-qubit systems and continuous-variable systems~\cite{hahn2022quantifying}.

In addition to the Wigner function, we also find that the  $l_p$-norm of the discrete-variable characteristic function (which corresponds to the coefficients of the generalized Pauli operators) exactly corresponds to the one of continuous-variable characteristic functions of GKP-encoded states. 
This provides a new interpretation of the stabilizer R\'enyi entropy---which is precisely defined by the $l_p$ norm of the Pauli coefficients---in terms of GKP encoding, and naturally extends it to all dimensions.

Furthermore, we obtain operational insights from our framework. 
We first show that our magic measure---which naturally emerges as the quantity corresponding to non-Gaussianity---is equipped with an operational meaning as a runtime of a classical simulation algorithm, extending the previous approaches of classical simulation based on quasiprobability distributions~\cite{pashayan2015estimating,rall2019simulation} to all dimensions. Furthermore, we introduce a simulation algorithm for GKP circuits that allows for all Gaussian unitaries as well as arbitrary qudit states encoded in GKP.
In addition, we apply our results to prove that the deterministic implementation of a logical non-Clifford operation with the same input and output systems in the GKP code subspace requires a non-Gaussian operation even in the limit of ideal GKP input state. 
Since ideal GKP states have unbounded non-Gaussianity, it is not a priori obvious that more non-Gaussianity is needed to apply a logical non-Clifford operation. 
Our result shows that this is actually the case in general, extending an observation for specific non-Clifford gates in multi-qubit systems~\cite{Yamasaki2020cost-reduced} to the general class of non-Clifford gates on all dimensions.

\section{Preliminaries}
\label{sec:preliminaries}

Here we briefly review the relevant formalism for discrete- and continuous-variable quantum computing.

\subsection{Discrete variables}
Qubits are ubiquitous in quantum information processing and are $d=2$ level systems.
Qudits are an intuitive generalizations to $d$ dimensions.
A general pure qudit state is defined as
\begin{align}
\ket{\psi}=\sum_{i=0}^{d-1} \alpha_i\ket{i}
\end{align}
with normalization condition $\sum_{i=0}^{d-1}\abs{\alpha_i}^2=1$ and $\ket{i}$ a computational basis state.
The Pauli group can be defined for arbitrary dimensions in analogy to the qubit case as
$\mathcal{P}_d=\qty{\omega_D^u X_d^v Z_d^w: v,w\in \mathds{Z}_d, u\in \mathds{Z}_D}$ where $\omega_d=e^{2\pi i/d}$ is the $d$\,th root of unity  and 
\begin{align}
    D=\begin{cases}
        d: & \text{for \textit{d} odd}\\
        2 d: & \text{for \textit{d} even}\\
    \end{cases}
    \label{eq:D definition}
\end{align}
with $\mathds{Z}_d$ being the integers modulo $d$.
The $d$ dimensional Pauli operators $Z_d, X_d$, sometimes also called clock and shift operators, are a way to generalize the qubit Pauli operators $Z_2, X_2$ and are defined as
\begin{align}
        X_d=& \sum_{j=0}^{d-1}\ket{j+1}\bra{j}\\
    Z_d=&\sum_{j=0}^{d-1}\omega_d^{j}\ket{j}\bra{j}
\end{align}
with the property $X_d^d= Z_d^d=\mathds{1}$ ~\cite{farinholt2014ideal}.

We use the generalized Pauli operators $X_d, Z_d$ to define the $d$-dimensional Heisenberg-Weyl operators as~\cite{gross2006hudson,gross2005finite}
\begin{align}
    P_d(a,b)= \omega_{d}^{\frac{1}{2} a b  } X_d^a Z_d^b
\end{align}
with $a,b\in \mathds{Z}_d$. Note here that for odd dimensions $\frac{1}{2}=2^{-1}$ is the multiplicative inverse on $\mathds{Z}_d$ while for even dimensions it means $ \omega_{d}^{\frac{1}{2}}= e^{i\pi/d}$, as there is no multiplicative inverse of 2.
Note that if we write $e^{i(\frac{\pi}{d})ab}$, the product $a b$ is taken in $\mathds{Z}$ rather than $\mathds{Z}_d$.
Then, the commutation relations are
\begin{align}
    P_d(a,b) P_d(c,d)= \omega_d^{(a,b)\Omega(c,d)^T}P_d(c,d) P_d(a,b),
\end{align}
where
\bal
\Omega=\begin{pmatrix}
0&-1\\1&0
\end{pmatrix}
\eal
is the symplectic form.

For $n$-qudit systems, the Heisenberg-Weyl operators are written by 
\bal
 P_d(\bm{u})= \otimes_{i=1}^n P_d(a_i,b_i)
 \label{eq:generalized Pauli definition}
\eal
with $\bm{u}=(\bm{a},\bm{b})\in\mathds{Z}_d^{2n}$, which satisfy the orthogonality relation  
\begin{align}
    \Tr\qty[P_d(\bm{u}) P_d^\dagger(\bm{v})]=d^n \delta_{\bm{u},\bm{v}}.
\end{align}
The $d$ dimensional $n$ qudit Clifford group is generated by the following unitary operations:
\bal
    R=& \sum_{j,s=0}^{d-1}\omega_d^{js}\ket{s}\bra{j}\\
    P=& \sum_{j=0}^{d-1}\omega_d^{j^2/2} (\omega_D \omega_{2d}^{-1})^{-j}\ket{j}\bra{j}\\
    \text{SUM}=&\sum_{i,j=0}^{d-1}         \ket{i}\bra{i}\otimes\ket{i+j } \bra{j}.
\eal
For $d=2$, these operators reduce to the Hadamard, Phase, and CNOT gate~\cite{farinholt2014ideal} respectively.
A Clifford unitary $ U_C$ acts on the Heisenberg-Weyl operator in a simple way
\begin{align}
U_C P_d(\bm{u})U_C^{\dagger}= e^{i\phi }P_d(S\bm{u})    
\end{align} 
where $S \in \text{SP}(2n,\mathds{Z}_D)$ is a symplectic matrix~\cite{farinholt2014ideal} associated with the Clifford unitary $U_C$,  and $e^{i\phi}$ is some phase factor.

Using the discrete Heisenberg-Weyl operators, we define the characteristic function~\cite{gross2006hudson}
\begin{align}
    \chi_{\rho }^{\DV}(\bm{u})=d^{-n}\Tr\qty[\rho P_d(\bm{u})^\dagger].
\end{align}
Odd-dimensional systems allow for a simple way to define the discrete Wigner function $W_\rho^\DV:\mathds{Z}_d^{2n}\to \mathds{R}$ as the discrete symplectic Fourier transform of the characteristic function
\begin{align}
     W_{\rho }^{\DV}(\bm{u})&=d^{-n} \sum_{\bm{v}\in \mathds{Z}_d^{2n}} \omega_d^{-\bm{u}^T\Omega_n\bm{v}}\chi_{\rho }^\DV(\bm{v})  \\ 
     &=d^{-n}\Tr\qty[A(\bm{u})\rho  ],
     \label{eq:discrete Wigner definition}
\end{align}
where $\Omega_n$ now takes the form 
\bal
 \Omega = \begin{pmatrix}
 0 & -\mathds{1}_n \\
 \mathds{1}_n & 0 
     \end{pmatrix}
\eal
and $\mathds{1}_n$ is the $n\times n$ identity matrix. 
The phase space point operator in Eq.~\eqref{eq:discrete Wigner definition} can be written as
\bal
    A(\bm{u})
    &= d^{-n} \sum_{\bm{v}\in\mathds{Z}_d^{2n}}\omega_d^{-\bm{u}^T\Omega_n\bm{v}}P_d(\bm{v})^\dagger.
    \label{eq:phase space point operator definition}
\eal
The discrete Wigner function $W_{\rho }(\bm{u})$ in odd dimensions has many useful properties. It is covariant under Clifford unitaries $U_C$ meaning that
\begin{align}
\label{eq:Wigner}
     W_{U_C\rho U_C^{\dagger}}^{\DV}(\bm{u}) =  W_{\rho }^{\DV}(S\bm{u})
\end{align}
with $S \in \text{Sp}(2n,\mathds{Z}_d)$ being the symplectic matrix associated with the Clifford unitary $U_C$.
Although it is a quasiprobability distribution and is thus not necessarily positive, $W_{\rho }(\bm{u})$ yields a valid probability distribution for an arbitrary stabilizer state. 
A pure stabilizer state is a state that is uniquely defined by a set of $d^n$ commuting Heisenberg-Weyl operators --the stabilizer group-- of which the state is in the $+1$ eigenspace of all operators in the stabilizer group. A mixed stabilizer state is a convex sum of pure stabilizer states. 
This property was utilized to introduce a magic measure--- a measure for non-stabilizerness--- given by the $l_1$-norm of the discrete Wigner function $W_\rho^\DV$, as
\begin{align}
    \|W_\rho^\DV\|_1= \sum_{\bm{u}}\abs{ W_{\rho }^{\DV}(\bm{u})}.
    \label{eq:negativity Wigner definition}
\end{align}
This quantity is therefore the  negativity of the Wigner function Eq.~\eqref{eq:Wigner}~\cite{veitch2014resource}.
This is a monotone under Clifford operations in the sense of resource theories~\cite{Chitambar2019quantum}, i.e. it is monotonically non-increasing under stabilizer protocols~\cite{veitch2014resource}, which consist of (1) Clifford unitaries (2) composition with stabilizer states (3) Pauli measurements (4) partial trace (5) the above operations conditioned on the outcomes of Pauli measurements or classical randomness. Moreover, it is faithful for pure states, i.e.,  $\|W_\psi^\DV\|_1 = 1$ if and only if $\psi$ is a stabilizer state for an arbitrary pure state $\psi$~\cite{gross2006hudson}.
 
In the following, we also consider the $l_p$-norm of a function $f:\mathds{Z}_d^{2n}\to\mathds{C}$ defined by
\bal
 \|f\|_p = \left(\sum_{\bm{u}\in\mathds{Z}_d^{2n}} |f(\bm{u})|^p\right)^{1/p}
\eal
for a real number $p>0$, therefore yielding a generalization of the negativity.

\subsection{Continuous variables}
The continuous-variable paradigm in quantum information processing uses infinite dimensional systems.
The central observables in these systems are commonly called position $Q$ and momentum $P$ and that fulfill the canonical commutation relations
\begin{align}
    \qty[Q,P]=i.
\end{align}
These operators have continuous spectra, which is why this type of quantum information processing is often called \emph{continuous variables}, in contrast to the \emph{discrete-variable} systems described in the previous section.
For a $n$-mode system, the infinite-dimensional Heisenberg-Weyl operators, also known as displacement operators, are defined by
\begin{align}
    D(\bm{r})= \prod_{j=1}^n e^{i r_{p_j} r_{q_j}/2} e^{-i r_{q_j} P_j} e^{i r_{p_j} Q_j}
\end{align}
where $\bm{r}=(r_{p_1},...,r_{p_n},r_{q_1},...,r_{q_n})= (\bm{r_p},\bm{r_q})$ and $Q_j$, $P_j$ are position and momentum operators for the $j$\,th mode. 
Displacement operators fulfill the commutation relation
\begin{align}
     D(\bm{r})  D(\bm{r'})= e^{-i \bm{r}^T\Omega_n\bm{r'}} D(\bm{r'}) D(\bm{r}).
\end{align}

Similarly to the discrete case, we can use the continuous Heisenberg-Weyl operators to define the characteristic function
\begin{align}
    \chi_{\rho }^{\CV}(\bm{r})=\Tr\qty[\rho D(-\bm{r})]
\end{align}
 and the Wigner function as its symplectic Fourier transform
\bal
    W_{\rho}^{\CV} (\bm{r})&=\frac{1}{(2\pi)^n} \int \dd\bm{r'} e^{i \bm{r}\Omega_n \bm{r'}}     \chi_{\rho }(\bm{r'})   \\ 
    &= \frac{1}{(2\pi)^n} \int_{-\infty}^\infty \dd^n\bm{x} e^{i\bm{r_p} \bm{x}} \bra{Q=\bm{r_q} +\frac{\bm{x}}{2}} \rho\ket{Q=\bm{r_q}-\frac{\bm{x}}{2}}.
\eal
The continuous Wigner function is a quasi-probability distribution and the 
Wigner negativity~\cite{Kenfack2004negativity}
\begin{align}
    \|W_\rho^\CV\|_1 = \int \dd{\bm{r}}\abs{W_{\rho }^{\CV}(\bm{r})} 
    \label{eq:Wigner negativity definition}
\end{align}
can be used as a valid meausure for non-Gaussianity~\cite{takagi2018convex,albarelli2018resource}. 
Similarly to the case of discrete variables, we also consider a $l_p$-norm for a function $f:\mathds{R}^{2n}\to\mathds{C}$ defined by 
\bal
 \|f\|_p = \left(\int d\bm{r} |f(\bm{r})|^p\right)^{1/p},
\eal
which for $p = 1$ gives back the Wigner negativity \eqref{eq:Wigner negativity definition} as the $l_1$-norm of the Wigner function.

A family of states playing a major role in this work are the Gottesman-Kitaev-Preskill (GKP) states~\cite{gottesman2001encoding}. 
They were originally introduced as error correction codes for bosonic quantum systems.
In this work, we employ this encoding as a platform to map magic and non-Gaussian resources for general qudit dimensions, extending a prior result for multiqubit systems~\cite{hahn2022quantifying}.

In the following, we use a subscript to denote a continuous-variable state that encodes a discrete-variable state. For instance, $\rho_\GKP$ refers to a continuous-variable state that encodes the qudit state $\rho$ by the GKP encoding. 
The computational basis state $\ket{j}$ is encoded in the GKP code as an infinite superposition of position eigenstates as
\begin{align}
    \ket{j}_\GKP &= \sum_{s=-\infty}^\infty \ket{Q= \alpha (j+ds)},
\end{align}
which is analogously described by the Wigner function
\begin{equation}\begin{aligned}
    &W_{\dyad{j}_\GKP}^{\CV}(r_p,r_q) \\
    &=\frac{1}{2\pi}\int_{-\infty}^\infty \dd x e^{ir_px}   \psi^j\left(r_q+\frac{x}{2}\right)^*     \psi^j\left(r_q-\frac{x}{2}\right)\\
    &\propto\sum_{s,t=-\infty}^\infty (-1)^{st} \delta\left(r_p-\frac{\pi}{d\alpha}s\right) \delta\left(r_q-\alpha j - \frac{d\alpha}{2}t\right)
    \label{eq:Wigner GKP}
\end{aligned}\end{equation}
with $\alpha = \sqrt{\frac{2\pi}{d}}$, and where $\psi^j\left(x) = \langle Q = x \ket{j}_\GKP \right.$. 
A useful property of the GKP code is that all Clifford unitaries on the code subspace can be implemented using Gaussian unitaries.

As can be seen in Eq.~\eqref{eq:Wigner GKP}, the Wigner function of a GKP state consists of a collection of delta functions.
This comes with an unbounded negativity, which reflects the fact that an ideal GKP state is unnormalizable. 
Nevertheless, one can see that the delta peaks are periodically positioned, with the unit cell having the size $\sqrt{2d\pi}\times \sqrt{2d\pi}$.
This motivates us to consider the $l_p$-norm of a function $f$ considered for a unit cell, defined by 
\bal
 \|f\|_{p,\cell}\coloneqq \left(\int_\cell d\bm{r} |f(\bm{r})|^p\right)^{1/p}
\eal
where $\int_\cell$ refers to the integral over the domain restricted to a hypercube $r_{q_i}\in [0,\sqrt{2d\pi}), r_{p_i}\in [0,\sqrt{2d\pi})$ in the phase space. 

We also note that there is a subtlety when we compute $l_p$-norm of a function that involves a delta function. We describe the procedure to perform such integrals in Appendix~\ref{ap:p-norm}.

\section{Bridging magic and non-Gaussianity}
\label{sec:bridge}
In this section, we present our results that directly connect the resource content of the discrete-variable state with that of the continuous one, establishing a quantitative relation between magic and non-Gaussianity. 
\subsection{Via Wigner function} \label{subsec:Wigner properties}
The first path to connect magic and non-Gaussianity uses the continuous Wigner function.

In order to make this connection, we consider an operator basis for qudit systems.
For $l,m\in\mathds{Z}_{2d}$, let $O_{l,m}$ be an operator defined by 
\bal
 O_{l,m} = e^{-i\pi ml/d}M_l Z_d^m
 \label{eq:operator basis definition}
\eal
where 
\begin{align}
    M_l=\sum_{\substack{u,v \in \mathds{Z}_d \\ u+v =l \mod d}} \ket{u}\bra{v}.
    \label{eq: ml def}
\end{align}
This is proportional to a variant of the phase space point operator studied in Ref.~\cite{Miquel2002quantum}---we will later make this connection more apparent.  
The operator in \eqref{eq:operator basis definition} can easily be extended to $n$-qudit systems, where we define $O_{\bm{l},\bm{m}}=\bigotimes_{i=1}^n O_{l_i,m_i}$ for $\bm{l},\bm{m}\in\mathds{Z}_{2d}^{n}$. 
The operators $O_{\bm{l},\bm{m}}$ are involutions that are orthogonal in the Hilbert-Schmidt inner product and therefore form a basis. Furthermore, they are closed under Clifford unitaries. We will delve more into the properties of the operators $O_{\bm{l},\bm{m}}$ in the next subsection.

Let us now define the distribution 
\bal
x_\rho (\bm{l},\bm{m}) \coloneqq d^{-n}\Tr(O_{\bm{l},\bm{m}} \rho )
\label{eq:distribution with our operator basis}
\eal
which corresponds to the coefficients for $O_{\bm{l},\bm{m}}$ when expanding the state $\rho$ with this operator basis.
Although $\bm{l},\bm{m}$ are elements of $\mathds{Z}_{2d}^n$ in general, the operators $O_{\bm{l},\bm{m}}$, and correspondingly $x_\rho(\bm{l},\bm{m})$, can only gain a phase factor by a translation $l_i\to l_i+d$ and $m_i\to m_i+d$ for any $i=1,\dots,n$, and that the operators $\{O_{\bm{l},\bm{m}}\}_{l,m\in\mathds{Z}_d^n}$ form an operator basis of a $n$-qudit system.
We consider the $l_p$-norm for this distribution over the restricted domain $\bm{l},\bm{m}\in\mathds{Z}_d^n$, i.e.,  
\bal
 \|x_\rho\|_p =\left(\sum_{\bm{l},\bm{m}\in\mathds{Z}_d^n} |x_\rho(\bm{l},\bm{m})|^p \right)^{1/p}.
\eal
Its properties as a magic measure depend on whether the dimension of the discrete-variable systems is even or odd.
Let us make this point explicit for the case of $p=1$. We first introduce the following proposition:

\begin{proposition}\label{prop:measure_properties}
The quantity $\|x_\rho\|_1$ is a magic measure for pure states, i.e. it satisfies the following properties, in any dimensions:
\begin{enumerate}
    \item Invariance under Clifford unitaries $U_C$:\\ $\norm{x_{U_C\rho U_C^\dagger}}_1 = \norm{x_{\rho}}_1  $ 
    \item Multiplicativity: $\norm{x_{\rho\otimes \sigma}}_1 = \norm{x_{\rho}}_1   \norm{x_{\sigma}}_1 $
    \item \label{item:free minimum} Pure stabilizer states achieve the minimum value:\\ $\|x_\phi\|_1= 1$ for every pure stabilizer state $\phi$, and $\|x_\psi\|_1\geq 1$ for every pure state $\psi$.
\end{enumerate}
\end{proposition}

The first two properties directly follow from the properties of $O_{\bm{l},\bm{m}}$, which we will address in detail in Sec.\ref{sec:Operatorbasis}. 
We show the third property in Appendix~\ref{app:free minimum}.

For odd dimensions, $\|x_\rho\|_1$ coincides with the discrete Wigner negativity~\cite{gross2006hudson,veitch2014resource}. 
As such, it can be defined for general mixed states, and it satisfies the properties discussed below Eq.(\ref{eq:negativity Wigner definition}), namely monotonicity under general Clifford protocols, also including Pauli measurements, along the following property:

\textbf{Proposition 1.(continued)}
\textit{
\begin{enumerate}
    \setcounter{enumi}{3}
    \item Faithfulness: For a pure state $\phi$,  $\|x_\phi\|_1=1$ if and only if $\phi$ is a stabilizer state.
\end{enumerate}
}

On the other hand, for the case of even dimensions, $\|x_\rho\|_1$, and more generally $\|x_\rho\|_p$, do not reduce to known magic measures in general. As a matter of fact, the definition of a discrete Wigner function in even dimensions is more challenging and involves expanding the set of phase-space point operators to an over-complete basis~\cite{raussendorf2020phase}.
In addition, the phase space point operators always have a unit trace~\cite{gross2006hudson, raussendorf2020phase}, while it is not the case for $O_{l,m}$ in even dimensions, indicating the subtlety of connecting it to Wigner functions.  
However, for the special case of $d=2$, we will show in Sec.\ref{sec:char} that $\|x_\rho\|_p$ is proportional to the stabilizer R\'enyi entropy for $\alpha = p/2$~\cite{leone2022stabilizer}, as the operators $O_{l,m}$ reduce to the Pauli operators. As a consequence, we recover faithfulness of $\|x_\psi\|_p$ for pure multi-qubit states in arbitrary $p$. In addition, it has recently been shown that the stabilizer R\'{e}nyi entropy with $\alpha\geq 2$ satisfies the monotonicity under stabilizer protocols~\cite{Leone2024stabilizer}, implying the same property for $\|x_\psi\|_p$ in qubit systems. Stabilizer protocols include Clifford unitaries, computational basis measurements, addition of computational basis states and feed-forward of the previous operations on measurement outcomes and classical randomness.

In summary, the major dividing line between even and odd dimensions for the above magic measures is the question of monotonicity under Pauli measurement and the inclusion of mixed states. Broadly speaking, it is an important problem to find computable magic measures for multi-qudit systems that are non-increasing under Pauli measurements. Beyond the stabilizer R\'{e}nyi entropy with $\alpha\geq 2$, one such measure is known as stabilizer nullity~\cite{beverland2020lower}, but it is a highly discontinuous measure unstable under an infinitesimally small perturbation. Finding such other computable measures with full monotonicity will make an interesting future direction.

Having introduced the main ingredients, we are ready to present the main result. We provide a fundamental and quantitative relation between magic and non-Gaussianity, both of which are central quantum resources in the major operational frameworks for quantum computing.
We do so by directly connecting the magic measure defined above to the amount of Wigner negativity of the continuous-variable state encoding the discrete-variable state.
More precisely, we connect the $l_p$-norm of the continuous Wigner function of a qudit encoded in GKP with the $l_p$ norm of the distribution defined in Eq.~\eqref{eq:distribution with our operator basis}.

\begin{theorem}\label{thm:CV-DV connection}
For an $n$-qudit state $\rho$ on a $d^n$-dimensional space and for an arbitrary real number $p>0$, it holds that
  \bal
   d^{n(1-1/p)}\|x_{\rho }\|_p = \frac{\|W^\CV_{\rho _\GKP}\|_{p,\cell}}{\|W^\CV_{\STAB_n, \GKP}\|_{p,\cell}},
   \label{eq:Wigner connection general}
  \eal
where 
\bal
\|W^\CV_{\STAB_n, \GKP}\|_{p,\cell}&\coloneqq \|W^\CV_{\phi_\GKP}\|_{p,\cell}\\
&=(4d)^{n/p}/(8\pi d)^{n/2}
\eal
is a quantity that takes the same value for every $n$-qudit pure stabilizer state $\phi$.
When $d$ is odd, we further have 
\bal
d^{n(1-1/p)}\|x_{\rho }\|_p=d^{n(1-1/p)}\|W^\DV_{\rho }\|_p = \frac{\|W^\CV_{\rho _\GKP}\|_{p,\cell}}{\|W^\CV_{\STAB_n, \GKP}\|_{p,\cell}}.
\label{eq:Wigner connection odd}
\eal

\end{theorem}

We prove Theorem~\ref{thm:CV-DV connection} later in this section using the following general relation between the continuous-variable Wigner function of the GKP state encoding a sate and the corresponding discrete Wigner function of the encoded state, which may be of interest on its own.

The peculiar property of the Wigner function of GKP states is that it comes with an atomic form, where the Dirac distribution has disjoint support
\begin{equation}\begin{aligned}
    &W^\CV_{\rho _\GKP}(\bm{r})\\
    &=\frac{\sqrt{d}^n}{\sqrt{8\pi}^n}\sum_{\bm{l},\bm{m}} c_{\rho_\GKP}(\bm{l},\bm{m})\delta\qty(\bm{r_p}-\bm{m}\sqrt{\frac{\pi}{2d}})\delta\qty(\bm{r_q}-\bm{l}\sqrt{\frac{\pi}{2d}})
    \label{eq:atomic}
\end{aligned}\end{equation}
where $c_{\rho_\GKP}(\bm{l},\bm{m})$ is a coefficient serving as a weight for each peak in the Wigner function of a GKP state $\rho_\GKP$.
We show how to derive Eq.~(\ref{eq:atomic}) from~(\ref{eq:Wigner GKP}) in Appendix~\ref{ap:atomic}.
This Wigner function forms a lattice, so we restrict it to one unit cell and focus on $l_i,m_i\in [0,2d-1]$ or equivalently $l_i,m_i \in \mathds{Z}_{2d}$ for each $i=1,\dots,n$.

The following result shows that the weight $c_{\rho_\GKP}(\bm{l},\bm{m})$ in the domain $\bm{l},\bm{m}\in\mathds{Z}_{2d}^n$ exactly coincides with the distribution defined in Eq.~\eqref{eq:distribution with our operator basis}.

\begin{proposition}\label{prop:coefficients relation}
For $\bm{l},\bm{m}\in\mathds{Z}_{2d}^n$, it holds that
\begin{align}
   c_{\rho_\GKP}(\bm{l},\bm{m}) =  x_\rho(\bm{l},\bm{m}).
\end{align}
\end{proposition}

The proof of Proposition~\ref{prop:coefficients relation} can be found in Appendix~\ref{ap:atomic}. 
This establishes the fundamental relation between discrete-variable and continuous-variable representations of an arbitrary state $\rho$.
As we will see later in this section, for odd dimensions Proposition~\ref{prop:coefficients relation} directly connects the discrete 
Wigner function of an arbitrary state $\rho$ and the continuous Wigner function of the GKP state that encodes $\rho$.

Theorem~\ref{thm:CV-DV connection} relates the $l_p$-norm of GKP Wigner functions to the $l_p$-norm of $x_\rho $ for a discrete-variable state $\rho $, which quantifies the magicness in $\rho $, and therefore establishes a direct connection 
between the Wigner negativity of a qudit encoded in GKP and a finite-dimensional magic measure.
The case of $p=1$ is particularly insightful. 
In this case, the 1-norm of the continuous-variable Wigner function coincides with the continuous-variable Wigner negativity, which is known to be a valid measure of non-Gaussianity~\cite{takagi2018convex, albarelli2018resource}. 
The quantity in Theorem~\ref{thm:CV-DV connection} is then the amount of non-Gaussianity renormalized by the negativity of the GKP states encoding stabilizer states.
This renormalization is necessary, as even stabilizer states encoded in GKP have non-zero continuous Wigner negativity.
In light of Theorem~\ref{thm:CV-DV connection}, the properties in Proposition \ref{prop:measure_properties} can be equivalently proven by leveraging on the properties of continuous Wigner function.

\subsubsection{Properties of operator basis}
\label{sec:Operatorbasis}

Since the operators $O_{\bm{l},\bm{m}}$ defined in Eq.~\eqref{eq:operator basis definition} are of interest in their own way and play a central role in connecting the continuous-variable and discrete-variable worlds as can be seen in Theorem~\ref{thm:CV-DV connection} and Proposition~\ref{prop:coefficients relation}, let us investigate their properties. 
A proof of the properties outlined here can be found in Appendix~\ref{ap:basis}.

Let us consider a single-qudit case, as the extension to multi-qudit case is straightforward. 
One can see that for $l,m\in\mathds{Z}_d$, the operator $O_{l,m}$ can also be written as
\bal
 O_{l,m} &= e^{-i\pi ml/d} R X_d^{-l} Z_d^m
 \label{eq:operator basis with reflection}
\eal
where $R\ket{j} = \ket{-j}$ for $j\in\mathds{Z}_d$ is the reflection operator for a computational basis state $\ket{j}$ with $\ket{-j}\coloneqq \ket{d-j}$. 
The form in \eqref{eq:operator basis with reflection} implies that our operator basis $\{O_{l,m}\}_{l,m}$ is equivalent to the phase space point operator introduced in Ref.~\cite{Miquel2002quantum} up to a constant factor.  
In the case of odd $d$, $X_d^{-l/2}$ and $Z_d^{m/2}$ are well defined (as $\mathds{Z}_d$ contains multiplicative inverse of 2 and thus $-l/2, m/2\in\mathds{Z}_d$), and the relation between $O_{l,m}$ and the phase space point operator \eqref{eq:phase space point operator definition} becomes apparent as 
\bal
O_{l,m} &= e^{-i\pi ml/d} \omega_d^{\frac{lm}{2}}  X_d^{l/2} Z_d^{-m/2} R Z_d^{m/2} X_d^{-l/2} \\
& = (-1)^{ml} P_d\left(\frac{l}{2},-\frac{m}{2}\right) R P_d\left(\frac{l}{2},-\frac{m}{2}\right)^\dagger\\
& = (-1)^{ml} A\left(\frac{l}{2},-\frac{m}{2}\right)
\eal
where we used $X_d^a R = R X_d^{-a}$ and $Z_d^b R = R Z_d^{-b}$ for every $a,b\in\mathds{Z}_d$. 
We discuss the relation between $O_{l,m}$ and the phase space point operator in detail in Appendix~\ref{ap:basis}.

In general, $l,m\in \mathds{Z}_{2d}$ are defined over $\mod 2d$.
However, for most applications, one can restrict to $\mathbb{Z}_d$.
For a value above $d$, the operators are periodic in $d$
\begin{align}
    M_l&=M_{l+d}\\
    Z_d^m &= Z_d^{m+d}
\end{align}
but can have different phases
\bal
   O_{l+d,m}&= (-1)^{m} O_{l,m}\\
   O_{l,m+d}&= (-1)^{l} O_{l,m}\\
    O_{l+d,m+d}&= (-1)^{l+m+d} O_{l,m}.
    \label{eq:phase shift rule}
\eal
Therefore, if one is only interested in the operators independent of the sign, one can restrict the domain of $l,m$.

In general, $O_{l,m}$ and $Z_d$ are unitary, and $M_l$ is Hermitian. Thus, it holds that
\begin{align}
    O_{l,m} O_{l,m}^{\dagger} =\mathds{1}.
\end{align}
The operator $O_{l,m}$ is also Hermitian $O=O^{\dagger}$
and thus 
\begin{align}
    O_{l,m}O_{l,m}=\mathds{1},
\end{align}
implying that the spectrum is $\pm 1$.

These operators are orthogonal in the sense of the Hilbert-Schmidt inner product
\begin{align}
     \Tr \qty(O_{l,m} O_{l'm'}) = \delta_{mm'}\delta_{ll'}d.
\end{align}
Furthermore, the action of Clifford unitaries on the operators $O_{\bm{l},\bm{m}}$ is equivalent to a symplectic linear transformation on the coordinates $(\bm{l},\bm{m})$ and constant shifts. We show this in Appendix~\ref{ap:cliffordcov}.
For $d=2$, one recovers the standard Pauli operators.
Therefore, $O_{l,m}$ can be seen as a Hermitian generalization of the Pauli operators to arbitrary dimensions.

$M_l$ contains $d$  $1$s for any dimension, but they behave differently for even and odd dimensions. 
Whether the operator is traceless for even dimensions depends on whether $l$ is even or odd---$M_l$ is traceless for odd $l$, while it has trace 2 for even $l$.
For odd dimensions, the matrices $M_l$ have trace $1$.

Using the properties of $M_l$ and the known properties of $Z_d$, we can now give a summary of the properties of $O_{l,m}$
\begin{align}
    \Tr[O_{l,m}] &=\begin{cases}
    1+(-1)^m & d,l\text{ even}\\
    0 & d \text{ even}, l \text{ odd}\\
    (-1)^{ml} & d \text{ odd}
    \end{cases}\\
    O_{l,m}&=O_{l,m}^{\dagger}\\
    O_{l,m}^2&=\mathds{1}\\
      \Tr \qty(O_{l,m} O_{l'm'}) &= \delta_{mm'}\delta_{ll'}d.
\end{align}
The properties above extend straightforwardly to the case of many qudits.

As we have seen that the operators are orthogonal under the Hilbert-Schmidt norm and form a basis, we can expand operators in that basis. We 
restrict to $\bm{l},\bm{m}\in \mathds{Z}_d^n$, since the operators for other $\bm{l},\bm{m}$ are the same modulo a potentially different sign that can be absorbed in the coefficients.

We can represent every multi-qudits quantum state in basis the operators $O_{\bm{l},\bm{m}}$ such that
\bal
    \rho &=\sum_{\bm{m},\bm{l}} \Tr\qty[\rho \frac{O_{\bm{l},\bm{m}}}{d^n} ]O_{\bm{l},\bm{m}}\\
    &=\sum_{\bm{l},\bm{m}} x_{\rho }(\bm{l},\bm{m})O_{\bm{l},\bm{m}}.
\eal
Given by the spectrum of $O_{\bm{l},\bm{m}}$, we can bound the value of the coefficients
\begin{align}
        -\frac{1}{d^n}  &\leq x_{\rho }(\bm{l},\bm{m})\leq \frac{1}{d^n}.
\end{align}
We can further bound $x_{\rho }(\bm{l},\bm{m}) $ using $\Tr(\rho)=1$ as 
\bal
   \Tr\qty( \rho  )&=\sum_{\bm{l},\bm{m}}x_{\rho }(\bm{l},\bm{m}) \Tr\qty(O_{\bm{l},\bm{m}})\\
   &=\begin{cases}
        \sum_{\bm{l},\bm{m}} (-1)^{\bm{l}\cdot\bm{m}}  x_{\rho }(\bm{l},\bm{m})=1&\text{odd}\\
          \sum_{\bm{l},\bm{m}:\text{even}} 2^n x_{\rho }(\bm{l},\bm{m})=1&\text{even}
    \end{cases}
\eal
where $\bm{l}\cdot\bm{m}=\sum_{i=1}^nl_{i}m_{i}$.

Interestingly, we find the following characterization of pure stabilizer states, which we prove in Appendix~\ref{ap:stab}.
\begin{proposition}\label{pro:pure stabilizer state}
 For an arbitrary pure stabilizer state $\phi$, $x_{\phi }(\bm{l},\bm{m})$ has the same magnitude $\abs{x_{\phi }(\bm{l},\bm{m})}=\frac{1}{d^n}$ with $d^n$ non-zero coefficients over $\bm{l},\bm{m}\in\mathds{Z}_d^n$.
\end{proposition}
We note that, because of the property in Eq.~\eqref{eq:phase shift rule}, the non-zero coefficients in the larger domain $\bm{m},\bm{l}\in\mathds{Z}_{2d}^n$ for a pure stabilizer state still solely takes the value $d^{-n}$, and the number of non-zero coefficients increases to $(4d)^n$.

We see that $x_\rho(\bm{l},\bm{m})$ is not directly a quasi-probability distribution but can easily be modified to be one for odd dimensions.
In the case of odd dimensions, we can show a direct connection between the operators $O_{l,m} $ and the phase space operators $A(a_1,a_2)$.
The precise connection is
\begin{align}
 O_{2a_1,-2a_2}=A(a_1,a_2),
\end{align}
where now the index $2a_1, -2a_2 \in \mathbb{Z}_{2d}$ go over numbers $\mod 2d$, to get the sign correct.
We can make the connection even more explicit by correcting the trace of $O_{l,m}$ and remembering that otherwise the phase space operators and  $O_{l,m}$ are reordered versions of each other
\begin{align}
    (-1)^{a_1a_2}O_{a_1,a_2}= A(\sigma[a_1,a_2]^T)
    \label{eq:phase point permutation}
\end{align}
where $\sigma$ is some permutation matrix over $\mathds{Z}_{d}^2$.
This result directly connects the discrete Wigner function $W_\rho^\DV$ with the distribution $x_\rho$ via
\begin{align}
    x_\rho(\bm{a_1},\bm{a_2})=(-1)^{\bm{a_1}\cdot \bm{a_2}} W_\rho^\DV(\sigma [\bm{a_1},\bm{a_2}]^T).
\end{align}

\subsubsection{Proof of Theorem~\ref{thm:CV-DV connection}}

We are now ready to show Theorem~\ref{thm:CV-DV connection}.
\begin{proof}[Proof of Theorem~\ref{thm:CV-DV connection}]
Using Propositions~\ref{prop:coefficients relation} and \ref{pro:pure stabilizer state}, we get for an arbitrary pure stabilizer state $\phi$ that 
\bal
\|W^\CV_{\phi_\GKP}\|_{p,\cell}&= \left[(4d)^n \left\{\left(\frac{d}{8\pi}\right)^{n/2} d^{-n}\right\}^p\right]^{1/p}\\
&= \frac{(4d)^{n/p}}{(8\pi d)^{n/2}}.
\eal
Proposition~\ref{prop:coefficients relation} gives 
\bal
 \|W_{\rho_\GKP}^\CV\|_{p,\cell} &= \left[\sum_{\bm{l},\bm{m}\in\mathds{Z}_{2d}^n} \left\{\left(\frac{d}{8\pi}\right)^{n/2} |x_\rho(\bm{l},\bm{m})|\right\}^p\right]^{1/p}\\
 &= \left[4^n\sum_{\bm{l},\bm{m}\in\mathds{Z}_{d}^n} \left\{\left(\frac{d}{8\pi}\right)^{n/2} |x_\rho(\bm{l},\bm{m})|\right\}^p\right]^{1/p}\\
  &= 4^{n/p}\left(\frac{d}{8\pi}\right)^{n/2}\|x_\rho\|_p\\
  & = d^{n(1-1/p)}\|W^\CV_{\phi_\GKP}\|_{p,\cell} \|x_\rho\|_p,
\eal 
which shows Eq.~\eqref{eq:Wigner connection general}.

For odd-dimensional cases, we can employ Eq.~\eqref{eq:phase point permutation} to get 
\bal
    \rho  &= \sum_{\bm{a}_1,\bm{a}_2} W_{\rho }^{\DV}(\bm{a}_1,\bm{a}_2) A(\bm{a}_1,\bm{a}_2) \\
    &= \sum_{\bm{a}_1,\bm{a}_2} x_{\rho }(\bm{a}_1,\bm{a}_2)O_{\bm{a}_1,\bm{a}_2}\\
    &=\sum_{\bm{a}_1,\bm{a}_2} x_{\rho }(\bm{a}_1,\bm{a}_2) (-1)^{\bm{a}_1\cdot\bm{a}_2} A(\sigma[\bm{a}_1,\bm{a}_2]^T).
\eal
This gives in particular
\bal
    \norm{W_\rho^\DV}_p&=\left(\sum_{\bm{a}_1,\bm{a}_2} \abs{W_{\rho }^{\DV}(\bm{a}_1,\bm{a}_2)}^p\right)^{1/p}\\
    &=\left(\sum_{\bm{l},\bm{m}}\abs{(-1)^{\bm{l}\cdot\bm{m}}x_{\rho }(\bm{l},\bm{m})}^p\right)^{1/p}\\
    &=\|x_\rho\|_p.
\eal
This, together with Eq.~\eqref{eq:Wigner connection general}, shows Eq.~\eqref{eq:Wigner connection odd}, completing the proof.
\end{proof}

This concludes the section on establishing a connection between magic and non-Gaussianity with Wigner functions.

\subsection{Via characteristic function}
\label{sec:char}
In this section, we use the formalism of characteristic functions to establish a connection between magic and non-Gaussianity similar to the one found in the previous section.

The characteristic function of a 
qudit state $\rho = \sum_{\bm{u},\bm{v} \in \mathds{Z}_d^n} \rho _{\bm{u},\bm{v}}\ket{\bm{u}}\bra{\bm{v}}$ encoded in GKP can then be written as
\begin{equation}\begin{aligned}
    &\chi_{\rho _\GKP}^{\CV}(\bm{r})\\
    &=\sqrt{\frac{2\pi}{d}}\sum_{l,m=-\infty}^\infty \gamma_{\rho _\GKP}(l,m) \delta\qty(p-m\sqrt{\frac{2\pi}{d}})  \delta\qty(q-l\sqrt{\frac{2\pi}{d}}),
\end{aligned}\end{equation}
where $\gamma_{\rho _\GKP}(l,m)$ is a coefficient serving as a weight for
each distribution of the characteristic function of a GKP state $\rho _\GKP$.
We show the derivation in Appendix~\ref{ap:char_fct}.

The following result establishes the fundamental connection between the discrete characteristic function and the continuous-variable characteristic function of GKP states that encodes the discrete-variable state.

\begin{theorem}\label{thm:characteristic}
Let $\rho$ be an $n$-qudit state on a $d^n$-dimensional space.
For $\bm{l},\bm{m}\in\mathds{Z}_{2d}^n$, it holds that
\bal
  \gamma_{\rho _\GKP}(\bm{l},\bm{m}) = d^ne^{-i\pi \bm{l}\cdot \bm{m}/d} \omega_d^{ -\bm{l}\cdot \bm{m}/2}\chi^\DV_\rho(\bm{l},\bm{m})^*.
  \label{eq:characteristic function coefficients}
\eal

In particular, 
  \bal
  d^{n(1-1/p)}\|\chi^\DV_\rho \|_p = \frac{\|\chi^\CV_{\rho _\GKP}\|_{p,\cell}}{\|\chi^\CV_{\STAB,\GKP}\|_{p,\cell}}
  \eal
where 
\bal
\|\chi^\CV_{\STAB,\GKP}\|_{p,\cell} &\coloneqq \|\chi^\CV_{\phi_\GKP}\|_{p,\cell}\\
& = \left(\frac{2\pi}{d}\right)^{n/2}(4d)^{n/p}
\eal
is a quantity that takes the same value for every pure stabilizer state $\phi$.
\end{theorem}
The proof can be found in Appendix~\ref{ap:char_fct}.
This gives us a direct connection to the $\alpha-$stabilizer R\'enyi entropy defined for multi-qubits~\cite{leone2022stabilizer}, which has recently been shown to be a magic monotone under stabilizer protocols for $\alpha\geq 2$~\cite{Leone2024stabilizer}. 
Our result also provides an immediate generalization to all dimensions.
The natural extension of the  $\alpha-$stabilizer R\'enyi entropy to $n$-qudit state is
\begin{equation}\begin{aligned}
    &M_{\alpha}(\rho ) \\
    &= (1-\alpha)^{-1}\log\qty(d^{-n\alpha}\sum_{P \in \mathcal{P}_n^*} \abs{\Tr\qty(\rho P)}^{2\alpha})-n\log d\\
    &= \alpha(1-\alpha)^{-1}\log \norm{\Xi(\rho )}_\alpha-n\log d
\end{aligned}\end{equation}
where  $\mathcal{P}_n^*$ is the projective generalized Pauli (Heisenberg-Weyl) group which only contains $+1$ phase, $\Xi_{P}(\rho )=\frac{1}{d^n}\Tr\qty(\rho P)^2$ forms a probability distribution when $\rho $ is pure.

Thus, we see immediately by comparison that we can write all $\alpha$-stabilizer R\'enyi entropies with the $l_{2\alpha}$-norm of the continuous-variable characteristic function for the qudit state that the GKP state encodes.
Specifically, we have 
\bal
 M_\alpha(\rho ) &= \frac{2\alpha}{1-\alpha}\log\|\chi^\DV_\rho \|_{2\alpha}-n\log d\\
 & = \frac{2\alpha}{1-\alpha}\log\frac{\|\chi^\CV_{\rho _\GKP}\|_{2\alpha,\cell}}{\|\chi^\CV_{\STAB,\GKP}\|_{2\alpha,\cell}}-\frac{\alpha\, n\log d}{1-\alpha}
\eal
where in the second equality we used Theorem~\ref{thm:characteristic}.

We show the faithfulness property for the case $p=1$ in Ap.~\ref{ap:faith}. It is straight forward to see that the other properties include invariance under Clifford unitaries and multiplicativity.

\section{Simulation Algorithms}
\label{sec:simulator}
In this section, we provide simulation algorithms for qudit circuits that use magic measures based on the connections we established with the Wigner and characteristic functions.
We give the technical details in Appendix~\ref{ap:simulation}.
Furthermore, we provide a simulator for GKP circuits that surprisingly resembles the introduced qudit simulators.
The details can be found in Appendix~\ref{ap:Simulating GKP}.
This provides an operational interpretation of the magic measures as the simulation overhead incurred by the non-classical features.    

\subsection{Wigner function}
\label{Sec:sim_wig}
The magic measures that were inspired by the continuous Wigner function allow for an operational interpretation of the simulation cost of a quantum circuit.
 Pashayan \emph{et al.}~\cite{pashayan2015estimating} introduced a simulation algorithm for quasi-probability distributions that strictly resemble the discrete Wigner functions in odd dimensions.
These ideas were used to adapt the simulator to multi-qubit cases by Rall \emph{et al.}~\cite{rall2019simulation}.
Using our unified approach that works for all dimensions, we can extend the simulator by Pashayan \emph{et al.} to all dimensions and recover the simulator of Rall \emph{et al.} for $d=2$ (see Appendix~\ref{ap:simulation}).
Crucial ingredients to do this are the properties of the operators $O_{\bm{l},\bm{m}}$, especially how they transform under Clifford unitaries. Consult Appendix~\ref{ap:basis} for details.
The simulation time scales with the aggregated $l_1$-norm of the circuit, similarly as for the total forward Wigner negativity of Ref.~\cite{pashayan2015estimating}:

\begin{align}
    \mathcal{M}_\rightarrow =\norm{x_{\rho }}_1  \prod_{t=1}^T \max_{\bm{\lambda_t}} \norm{x_{U_t}(\bm{\lambda_t})}_1    \max_{\bm{\lambda_T}} \abs{x_{\Pi}(\bm{\lambda_T})}
\end{align}
where the maximum is taken over all trajectories and
\begin{align}
    x_{\rho }(\bm{\lambda})&=\Tr\qty(\rho  \frac{O_{\bm{\lambda}}}{d^n})\\
    x_{U}(\bm{\lambda'},\bm{\lambda})&=\Tr \qty(\frac{O_{\bm{\lambda'}}}{d^n }UO_{\bm{\lambda}}U^{\dagger}) \\
    x_{\Pi}(\bm{\lambda})&= \Tr\qty(\Pi O_{\bm{\lambda}}).
\end{align}
are the coefficients related to the input state $\rho $, the unitary evolution $U$ and the measurement effect $\Pi$ with $\lambda=(\bm{l},\bm{m})\in \mathbb{Z}_d^{2n}$.
The number of samples $K$ that achieves precision $\epsilon$ with a failure probability $p_f$ is given by 
\begin{align}
    K\geq 2 \mathcal{M}_\rightarrow^2 \frac{1}{\epsilon^2} \ln\qty(\frac{2}{p_f}).
\end{align}
This shows that the number of samples directly scales with the resourcefulness of the input state $\norm{x_{\rho }}_1 $ if we evolve using the free operations of our magic meausures like Clifford unitaries and measure in the computational basis.
This result gives an operational interpretation of the meausures discussed in this work, as already noted in the case for the discrete Wigner negativity~\cite{pashayan2015estimating}. 

A few comments on the difference between even and odd dimensional systems are in order.
 For odd dimensional systems, it holds $\norm{x_{\rho }}_1\geq 1 $, whereas it is possible for qubits that $\norm{x_{\rho }}_1 \leq 1$ for specific states which reduce the number of samples needed. These qubit states are discussed in~\cite{rall2019simulation} and are called hyperoctahedral states. 
 We show that this phenomenon exists in all even dimensions and call these states hyperpolyhedral states. See Appendix~\ref{ap:hyper} for details.
However, for pure states, it holds that $\norm{x_{\rho }}_1\geq 1 $.

Another interesting difference between even and odd dimensions when Pauli measurements are involved is that the simulator is less competitive in even dimensions.
The cost of the measurements is taken into account via the term
\begin{align}
    \max_{\bm{\lambda_T}} \abs{x_{\Pi}(\bm{\lambda_T})}.
\end{align}
Since for odd-dimensional systems $O_{\bm{l},\bm{m}}$ has trace 1, computational basis measurements do not increase the simulation time.
This is not the case for even dimensions. In this case, the measurements increase simulation time, as was noted by Rall \emph{et al.}~\cite{rall2019simulation}.

Let us assume that we would like to measure $k$-qudits of our $n$-qudit system in a computational basis state $\ket{\bm{i}}$. 
The measurement effect then is given as $\Pi=\mathds{1}_{n-k}\otimes \dyad{\bm{i}}$. Without loss of generality, assume the measurement of the state $\dyad{\bm{1}}$, the state with all measured qudit in the 1 state.
The expansion of a qudit in the operators $O_{l,m}$ is $\dyad{1}=\frac{1}{d}\sum_{i=1}^{d-1}O_{2,i}$.
The cost inferred from the measurement is then
\begin{align}
     &\max_{\bm{\lambda_T}} \abs{x_{\Pi}(\bm{\lambda_t})}= \max_{\bm{l},\bm{m}} \abs{\Tr\qty[O_{\bm{l},\bm{m}} \mathds{1}_{n-k}\otimes \dyad{\bm{1}} ]}\\
     = &\max_{\bm{l_{n-k}},\bm{m_{n-k}}}  \abs{\Tr \qty[ O_{\bm{l_{n-k}},\bm{m_{n-k}}} ]} \max_{\bm{l_{k}},\bm{m_{k}}} \abs{\Tr \qty[O_{\bm{l_k},\bm{m_k}} \dyad{\bm{1}} ]}.
\end{align}
The maximum trace $\abs{\Tr \qty[O_{\bm{l_k},\bm{m_k}} \dyad{\bm{1}} ]}$ is $1$ for both even and odd dimensions. However, for the first term there is a big difference between even and odd dimensions. For odd dimensions the trace of $O_{\bm{l},\bm{m}}$ is $\pm 1$, so unmeasured qudits do not add to the simulation cost in any way. This is not the case for even dimensions. In even dimensions the trace of a single qudit operator $O_{l,m}$ is either $0$ or $2$. 
Therefore, the maximum of the first term  $\abs{\Tr \qty[ O_{\bm{l_{n-k}},\bm{m_{n-k}}} ]}$ is $2^{n-k}$, and thus the number of unmeasured qudits increase the number of samples required exponentially.

However, this cost only plays a role if not all qudits are measured. Relevant examples where naturally all qudits are measured at the end of the computations are variational algorithms and sampling algorithms.
Also, the existence of hyperoctahedral states admits the simulation of noisy magic states at the cost of pure stabilizer states (or even less in some cases), which allows us to use the peculiar behavior of even-dimensional systems to our advantage.

\subsection{Characteristic function}
The same ideas can be used to construct a simulator that is based on characteristic functions. It will reduce to the simulator by Rall \emph{et al.}~\cite{rall2019simulation} for $d=2$.
The simulation cost scales with a resource quantified using magic measures based on the connection of the characteristic functions.

Instead of representing the quantum state in the basis of $O_{\bm{l},\bm{m}}$, we use the Heisenberg-Weyl operators $P_d(\bm{l},\bm{m})$ defined in Eq.~\eqref{eq:generalized Pauli definition} as our basis.
Since they are unitary and traceless (with the exception of the identity operator) a few small modification are in order.
A qudit state $\rho$ and its characteristic function $\chi^\DV_\rho$ can be written as
\bal
    \rho&= \frac{1}{d} \sum_{\bm{l},\bm{m}} \chi_\rho^\DV(\bm{l},\bm{m}) P_d(\bm{l},\bm{m})
\eal
with $ \chi_\rho^\DV(\bm{0},\bm{0})=1$, since $P_d(\bm{0},\bm{0})=\mathds{1}$.
Furthermore, since the density operator is Hermitian, it holds that
\bal
    \sum_{\bm{l},\bm{m}} \chi_\rho^\DV(\bm{l},\bm{m}) P_d(\bm{l},\bm{m}) &= \sum_{\bm{l},\bm{m}} \left[\chi_\rho^{\DV}(\bm{l},\bm{m})\right]^* P_d(\bm{l},\bm{m})^\dagger.
\eal
Therefore, many coefficients in the decomposition are redundant.
For Heisenberg-Weyl operators it holds that
\begin{align}
    P_d^\dagger (\bm{l},\bm{m})= \omega_d^{\bm{l}\cdot\bm{m}} P_d(-\bm{l},-\bm{m})
\end{align}
and therefore
\begin{align}
     \left[\chi_\rho^{\DV}(\bm{l},\bm{m})\right]^*=\omega_d^{\bm{l}\cdot\bm{m}}  \chi_\rho^\DV(-\bm{l},-\bm{m}).
\end{align}
This implies that we have only $\frac{d^2-1}{2}$ independent coefficients in the decomposition. 
Thus, we can only sample from the independent coefficients, since they are pairwise dependent.

The rest of the algorithm works equivalently. In particular, the simulation time scales with the aggregated $l_1$-norm:
\begin{align}
    \mathcal{M}_\rightarrow^{\chi} =\norm{ \chi_\rho^\DV}_1  \prod_{t=1} \max_{\bm{\lambda_t}} \norm{ \chi_{U_t}^\DV(\bm{\lambda_t})}_1    \max_{\bm{\lambda_T}} \abs{ \chi_\Pi^\DV(\bm{\lambda_T})}
\end{align}
with
\begin{align}
 \chi_\rho^\DV(\bm{\lambda})&= \Tr\qty(\rho \frac{P_d^\dagger(\bm{\lambda})}{d^n})\\\
 \chi_U^\DV(\bm{\lambda'},\bm{\lambda})&=\Tr\qty(\frac{P_d^\dagger(\bm{\lambda'})}{d^n} U P_d(\bm{\lambda})U^\dagger    )\\
     \chi_\Pi^\DV(\bm{\lambda'}) &= \Tr\qty(\Pi P_d(\bm{\lambda})).
\end{align}
The simulator behaves similarly to the previous one for even dimensions.
The same simulation time increase happens for unmeasured qudits in this case.

\subsection{Simulating GKP}\label{sec:simulation GKP}

The connection between the continuous Wigner function and the introduced discrete distributions can also be used to simulate the dynamics of GKP states. 
Here, we introduce a simulation algorithm for circuits that use ideal GKP codewords, Gaussian unitary operations and Gaussian measurements, i.e. homodyne detection.
The details can be found in Appendix~\ref{ap:Simulating GKP}. We remark that the works~\cite{calcluth2022efficient,calcluth2023vacuum} show a simulator that efficiently simulates GKP qubit stabilizer states with rational symplectic unitaries with displacements. 
On the other hand, our algorithm works beyond this restricted set of operations, encompassing all Gaussian unitaries and measurements.
Instead, our algorithm only weakly simulates the dynamics, while the one in \cite{calcluth2023vacuum} realizes strong simulation.

We want to simulate the run of a quantum circuit, but in contrast to the previous section, in the end we want to obtain a sample $\bm{x}$ of a homodyne measurement.
The sample $\bm{x}$ is drawn from the probability distribution of obtaining a certain measurement outcome $\bm{x}$ given by the Born probability
\begin{align}
    P(\bm{x})=\Tr[ \Pi_{\bm{x}} U_G \rho U_G^\dagger ]
\end{align}
where $\Pi_{\bm{x}}$ describes the measurement effect of obtaining measurement outcome $\bm{x}$.
However, a few comments are in order. First, GKP states are not quantum states, since they are non-normalizable. In consequence, if $\rho$ in the equation above is such an object, $P(x)$ is not a probability distribution, i.e. the integral over all measurement outcomes is not 1 but $\infty$. 
Even though such a quantity is not properly defined in a physical context, this idealized case can still give insights.
To differentiate between a probability $P(x)$ and the quantity we obtain by measuring an ideal GKP state, we call the latter $\Tilde{P}(x)$.

We will assume now an ideal GKP codeword encoding a multi-qudit state $\rho$. We then rewrite the ``probability'' we want to sample from as
\begin{align}
     \Tilde{P}(\bm{x})&= \int \dd \bm{r}\; W^{\CV}_{\Pi_{\bm{x}}}(\bm{r}) W^{\CV}_{U_G\rho_\GKP U_G^\dagger}(\bm{r})\\
     &= \frac{\sqrt{d}^n}{\sqrt{8\pi}^n} \sum_{\bm{u}=-\infty}^\infty x_{\rho}(\bm{u})  \int \dd \bm{p}\;\delta\left( [\bm{p},\bm{x}]^T+ S \bm{d}-\sqrt{\frac{\pi}{2d}} S \bm{u}\right),
\end{align}
where we used the atomic form defined in Eq.~\eqref{eq:atomic}.

We use the periodicity of  coefficients $ x_{\rho}(\bm{u})$ to simplify the expression further and obtain 
\begin{widetext}
  \begin{align}
     \Tilde{P}(\bm{x})=\frac{\sqrt{d}^n}{\sqrt{8\pi}^n} \sum_{\bm{u}\in \mathds{Z}_{2d}^{2n}} \text{sign}( x_{\rho}(\bm{u})  )\norm{x_{\rho}}_1 \frac{\abs{x_{\rho}(\bm{u})}}{\norm{x_{\rho}}_1}    \sum_{\bm{n}\in \mathds{Z}^{2n}}\delta\qty( \bm{x}+ \Tr_{\bm{p}}\qty[S \bm{d}-\sqrt{\frac{\pi}{2d}} S \bm{u}-\sqrt{2\pi d}S\bm{n}]).
\end{align}  
\end{widetext}
Thus to obtain a sample according to the equation above we first sample a $\bm{u}$ according to $\frac{\abs{x_{\rho}(\bm{u})}}{\norm{x_{\rho}}_1}$.
Then we sample uniformly a $\bm{n} \in \mathds{Z}_{2d}^{2n}$ and return $\bm{x} = \Tr_{\bm{p}}\qty[-S \bm{d}+\sqrt{\frac{\pi}{2d}} S \bm{u}+\sqrt{2\pi d}S\bm{n}]$ as the index of the measurement result.
Here another comment is in order. Sampling from the integers is not possible since the set is not closed. We can however use the form we derived by just investigating how the first unit cell of the GKP lattice evolves by setting $\bm{n}=0$.
Then the dynamics have the familiar form
\begin{widetext}
  \begin{align}
     \Tilde{P}(\bm{x})=\frac{\sqrt{d}^n}{\sqrt{8\pi}^n} \sum_{\bm{u}\in \mathds{Z}_{2d}^{2n}} \text{sign}( x_{\rho}(\bm{u})  )\norm{x_{\rho}}_1\frac{\abs{x_{\rho}(\bm{u})}}{\norm{x_{\rho}}_1}    \delta\qty( \bm{x}+ \Tr_{\bm{p}}\qty[S \bm{d}-\sqrt{\frac{\pi}{2d}} S \bm{u}])
\end{align}  
\end{widetext}
and the output sample depends only on the qudit distribution $ x_{\rho}$. Note the resemblance with the simulation algorithm presented in Sec.~\ref{Sec:sim_wig}.

\section{Magic needs non-Gaussianity}
\label{sec:nonc}

It has been known since the original GKP paper~\cite{gottesman2001encoding} that one can implement the logical $T$-gate and thus get an $H$-type magic state by using a cubic phase state or cubic interaction $e^{icQ^3}$. However, this is merely one possibility for implementing a non-Gaussian interaction, and this does not show the necessity of non-Gaussianity to implement a non-Clifford operation on the code subspace.
This is a widely held belief based on the correspondence between a pair of Pauli and displacement operators and that of Clifford and Gaussian operations, where displacement operators and Pauli operators are both Heisenberg-Weyl operators. 
However, this ``belief'' has not been proven in general, beyond specific cases in qubit systems~\cite{Yamasaki2020cost-reduced}, where the conversion between a computational basis state and the $H-$state was ruled out.
Indeed, GKP states have large Wigner negativity and thus \emph{a priori} additional non-Gaussianity may not be required, making the necessity of non-Gaussian operation to implement a non-Clifford operation nontrivial. 

Nevertheless, the results established above allow us to show that non-Gaussian operations are essential to implement non-stabilizer operations in the GKP code space. 
In fact, we find that the \emph{Gaussian protocols}~\cite{takagi2018convex}---a class of deterministic quantum channels larger than Gaussian operations, which also admits feed-forwarded Gaussian operations conditioned on the outcomes of Gaussian measurements---are not able to implement non-stabilizer operations in the GKP code space.
Importantly, Gaussian protocols include a gate teleportation circuit involving a Gaussian measurement and a feed-forwarded Gaussian unitary, which itself is not a Gaussian operation~\footnote{Here, we follow the standard definition where Gaussian operations are quantum channels that map Gaussian states to Gaussian states}.

\begin{theorem}
\label{th:clifford}
    Let $\Lambda$ be a quantum channel with $n$-qubit input and output. If there exists a pure stabilizer state $\phi$ and a pure non-stabilizer state $\psi$ such that $\Lambda(\phi)=\psi$, $\Lambda$ cannot be implemented in a GKP code space by a Gaussian protocol. Also, for a quantum channel $\Lambda$ with $n$-qudit input and output systems with odd local dimensions, the condition can be relaxed to the existence of a (potentially mixed) stabilizer state $\sigma$ and a state $\rho$ with $\|W_\rho^\DV\|_1>1$ such that $\Lambda(\sigma)=\rho$.
\end{theorem}
\begin{proof}
  Suppose that $\Lambda$ can be implemented in the GKP code space by a Gaussian protocol $\mG$, i.e., $\mG(\sigma_\GKP)=\rho_\GKP$ for qudit states $\sigma$ and $\rho$ such that $\Lambda(\sigma)=\rho$.
  Since $\|W_{\rho_\GKP}^\CV\|_{1,\cell}$ does not increase under Gaussian protocols~\cite{takagi2018convex,albarelli2018resource,hahn2022quantifying}, we get 
  \bal
   \|W_{\rho_\GKP}^\CV\|_{1,\cell} =\|W_{\mG(\sigma_\GKP)}^\CV\|_{1,\cell} \leq \|W_{\sigma_\GKP}^\CV\|_{1,\cell}.
   \label{eq:Wigner negativity monotonicity}
  \eal
Because of the assumption that $\rho$ is also an $n$-qudit state, Theorem~\ref{thm:CV-DV connection} and Eq.~\eqref{eq:Wigner negativity monotonicity} imply that 
\bal
 \|x_\sigma\|_1 &= \frac{\|W_{\sigma_\GKP}^\CV\|_{1,\cell}}{\|W_{\STAB_n, \GKP}^\CV\|_{1,\cell}}\\
 & \geq  \frac{\|W_{\rho_\GKP}^\CV\|_{1,\cell}}{\|W_{\STAB_n, \GKP}^\CV\|_{1,\cell}}\\
 & = \|x_\rho\|_1.
 \label{eq:non-Clifford operation inequality}
\eal  

Suppose that the input state $\sigma$ is a pure stabilizer state denoted by $\phi$ and the output state $\rho$ is a pure non-stabilizer state $\psi$. 
Since $\|x_\phi\|_1$ is faithful for pure states as shown in Sec.~\ref{subsec:Wigner properties}, i.e., for a pure state $\phi$, $\|x_\phi\|_1=1$ if and only if $\phi$ is a stabilizer state, we get $\|x_\phi\|_1=1$ and $\|x_\psi\|_1>1$.
This is a contradiction with \eqref{eq:non-Clifford operation inequality}, showing that such a channel $\Lambda$ cannot be implemented by a Gaussian protocol.

The statement for odd dimensions follows by the same argument using the relation \eqref{eq:Wigner connection odd}.

\end{proof}
An immediate consequence is that a Gaussian protocol cannot implement non-Clifford unitary gates deterministically. 
This does not contradict the protocol by Baragiola \emph{et al.}~\cite{baragiola2019all}, which requires many auxiliary GKP states---making the whole operation involving the preparation of such auxiliary states highly non-Gaussian---to apply a single non-Clifford gate. 
In addition, their protocol is probabilistic and, therefore, does not directly fall into the scope of our result, which is pertinent to deterministic operations.

We stress that the statement of Theorem~\ref{th:clifford} directly benefits from the connection between the DV and CV resources established in this work. In particular, the general connection between DV and CV resource measures in Theorem~\ref{thm:CV-DV connection} and the property of magic measure in Proposition~\ref{prop:measure_properties} play crucial roles in the proof of Theorem~\ref{th:clifford}.

\section{Conclusion and Outlook}
\label{sec:conclusion}

In this work, we establish a quantitative relationship between magic and non-Gaussianity through the Gottesman-Kitaev-Preskill encoding, which serve as the central resources for the key operational frameworks used to study quantum computation in both discrete and continuous variables.  Furthermore, our work provides a tool to analyse resources in the setting when qudit systems are encoded in CV systems.
We introduced a family of distributions for discrete-variable systems and showed that their $l_p$-norm exactly corresponds to that of the continuous Wigner function that encodes the same qudit states via GKP encoding. 
Notably, the discrete-variable distribution coincides with the discrete Wigner function for odd dimensions, allowing us to connect the negativity of Wigner functions of discrete and continuous variables for $p=1$.
More generally, our distributions allow for defining a magic meausure for all dimensions and extend the discrete Wigner negativity defined for odd dimensions and the stabilizer R\'enyi entropy defined for multi-qubit systems in a unified manner. 
Furthermore, we showed that the $l_p$-norm of the discrete-variable characteristic function corresponds to the characteristic function of a GKP state that encodes the same qudit state. This provides a new interpretation of the stabilizer R\'enyi entropy in terms of the GKP encoding and naturally extends it to all dimensions.
By employing this framework, we find an operational interpretation of the magic measures by introducing a classical simulation algorithm, where the run time scales with the magic content. The first algorithm is based on the magic measures connected to the continuous Wigner function and recovers the simulator in Ref.~\cite{pashayan2015estimating}, while the second one is based on the magic measures connected to the continuous characteristic function. 
Both algorithms give a strong operational interpretation to the magic measures we introduced.
Then we presented a weak simulation algorithm for ideal GKP circuits. This algorithm improves state of the art by allowing all Gaussian unitaries as well as all qudit states encoded in GKP. This algorithm bores many similarities to the qudit simulators we introduced before. 
We utilized our findings to demonstrate that achieving a deterministic implementation of a logical non-Clifford operation, with identical input and output dimensions within the GKP code subspace, necessitates a non-Gaussian operation, even when operating at the theoretical limit of ideal GKP state input. 
We conjecture that the faithfulness property holds independently of the dimension and not only for odd dimensional and multi-qubit systems.  An implication is that Theorem~\ref{th:clifford} would hold for all dimensions, where one could replace multi-qubit systems with any even-dimensional ones.
We leave that conjecture for future work.

Our framework offers a novel approach to analyze non-Gaussian and magic resources in a mutual way: tools developed for infinite-dimensional systems can be used to describe properties of finite-dimensional systems, and vice-versa. We hypothetise that further cross-fertilization between these two worlds is possible, allowing for investigating more properties in the light of our approach. 
Furthermore, we have seen that the magic measures defined in this work behave differently for even and odd dimensions. An interesting future direction is to further investigate the origin of this behavior. Finally, being able to investigate and see the dependence of the dimensionality could shed new light on the source of quantum speed-ups.

\emph{Note added.}---During the completion of this manuscript, a related independent work by Lingxuan Feng and Shunlong Luo~\cite{Feng2024connecting} was brought to our attention, where the authors found a complementary relation between the description of a single qudit state and continuous-variable Wigner function of the corresponding GKP state.

\begin{acknowledgments}
We thank Kaifeng Bu, Ulysse Chabaud, Cameron Calcluth, and Nicolas Menicucci for helpful discussions. G.F. acknowledges support from the Swedish Research Council (Vetenskapsrådet) through the project grant DAIQUIRI, as well as from the HORIZON-EIC-2022-PATHFINDERCHALLENGES-01 programme under Grant Agreement Number 101114899 (Veriqub). G.F., O.H. acknowledge support from the Knut and Alice Wallenberg Foundation through the Wallenberg Center for Quantum Technology (WACQT). R.T. acknowledges the support of JSPS KAKENHI Grant Number JP23K19028, JP24K16975, JST, CREST Grant Number JPMJCR23I3, Japan, and MEXT KAKENHI Grant-in-Aid for Transformative
Research Areas A ``Extreme Universe” Grant Number JP24H00943.
\end{acknowledgments}


\newpage
\appendix

\section{ $l_p$ norm and renormalization}
\label{ap:p-norm}
In this section, we formalize a way to compute the $l_p$ norm of the characteristic function as well as the Wigner function of GKP states.
 Note here that the $l_1$ norm of the Wigner function is the Wigner negativity.
 The $l_p$ norm is defined as
\begin{align}
    \norm{f}_p=\qty(\int \dd\bm{x}  \abs{f(\bm{x})}^p )^{\frac{1}{p}}.
\end{align}

We are interested in computing the $l_p$-norm of characteristic and Wigner functions of GKP states, 
so we deal with sums of Dirac distribution, where the distributions have disjoint support
\begin{equation}\begin{aligned}
  &\qty(\int_{-\infty}^\infty \dd{\bm{x}}\abs{ \sum_i f_i(\bm{x})  \delta\qty(\bm{x} - \bm{x}_i)}^p  )^\frac{1}{p}\\
  &=\qty(\sum_i \abs{f_i(\bm{x}_i )}^p \delta(0)^{p-1}   )^\frac{1}{p}= \qty(\sum_i \abs{f_i(\bm{x}_i )}^p  )^\frac{1}{p} \delta(0)^\frac{p-1}{p} .
\end{aligned}\end{equation}
This integral evaluates to the same Dirac distribution $\delta(0)^{(p-1)/p}$ for all GKP states, which will be canceled by dividing it by the $l_p$-norm for another GKP state as in Theorems~\ref{thm:CV-DV connection}~and~\ref{thm:characteristic}. 
Therefore, we will define the norm as
\begin{equation}\begin{aligned}
    &\qty(\int_{-\infty}^\infty \dd{\bm{x}}\abs{ \sum_i f_i(\bm{x})  \delta\qty(\bm{x} - \bm{x}_i)}^p  )^\frac{1}{p}\\ 
    &=\qty( \int_{-\infty}^\infty \dd{\bm{x}}\abs{ \sum_i f_i(\bm{x})}^p  \delta\qty(\bm{x} - \bm{x}_i) )^\frac{1}{p}\\
    &= \qty(\sum_i \abs{f_i(\bm{x}_i )}^p  )^\frac{1}{p}
\end{aligned}\end{equation}
as a kind of regularization.


\section{Proof of Property~\ref{item:free minimum} of $\|x_\rho\|_1$}
\label{app:free minimum}
In this section, we prove a property of $\|x_\rho\|_1$  among those listed in Sec.~\ref{sec:bridge}, specifically that $\|x_\phi\|_1=1$ for every pure stabilizer state $\phi$ and $\|x_\psi\|_1\geq 1$ for every pure state $\psi$.  

We have the requirement for a pure state $\psi$ that
\bal
    \Tr(\psi ^2)=\sum_{\bm{l},\bm{m}} x_{\psi}(\bm{l},\bm{m})^2 d^n&=1,
\eal
which implies
\bal
\sum_{\bm{l},\bm{m}}\Tr(O_{\bm{l},\bm{m}} \psi )^2=d^n.
\eal
Recalling that $O_{\bm{l},\bm{m}}$ has eigenvalues $\pm 1$, it holds that $|\Tr(O_{\bm{l},\bm{m}}\psi)|\leq 1,\ \forall \bm{l},\bm{m}$.  
This gives
\bal
 \sum_{\bm{l},\bm{m}}|\Tr(O_{\bm{l},\bm{m}} \psi )| \geq \sum_{\bm{l},\bm{m}}\Tr(O_{\bm{l},\bm{m}} \psi )^2=d^n,
\eal
showing $\|x_\psi\|_1\geq 1$ for every pure state $\psi$.

Let $\phi$ be an arbitrary pure stabilizer state. 
Proposition~\ref{pro:pure stabilizer state} ensures that $|x_\phi(\bm{l},\bm{m})|=1/d^n$ for $d^n$ elements, leading to 
\bal
 \|x_\phi\|_1 = \sum_{\bm{l},\bm{m}}|x_\phi(\bm{l},\bm{m})| = \frac{1}{d^n}\cdot d^n = 1,
\eal
completing the proof.

\section{Proof of Proposition~\ref{prop:coefficients relation}}
In this section, we derive the atomic form of an $n$-qudit state encoded in the Gottesman-Kitaev-Preskill (GKP) code. We call the representation atomic if each Dirac distribution with different support appears only once in the summation, thus all Dirac distributions are distinct. 
We start deriving the atomic form for one qudit encoded in GKP and then generalize it to $n$-qudit systems.
\label{ap:atomic}
\subsection{One Qudit}
For a single qudit, the Wigner function of a computational basis state $\dyad{j}$ encoded in GKP are  
\begin{equation}\begin{aligned}
    &W_{\dyad{j}_\GKP}^{\CV}(r_q,r_p)\\
    &\propto \sum_{s,t=-\infty}^\infty (-1)^{st} \delta(r_p-\frac{\pi}{d\alpha}s) \delta(r_q-\alpha j - \frac{d\alpha}{2}t).
\end{aligned}\end{equation}
In order to derive the Wigner function of an arbitrary qudit state $\rho = \sum_{u,v \in \mathds{Z}_d} \rho _{u,v}\ket{u}\bra{v}$  encoded in the GKP code, we expand our state in the computational basis
\begin{widetext}
\begin{align}
\label{eq:GKP1}
    &W_{\rho_\GKP }^{\CV}(r_p,r_q) = \sum_{u,v \in \mathds{Z}_d} \rho _{uv}   \frac{1}{2\pi}\int_{-\infty}^\infty \dd{x} e^{i r_p x}\qty[\sum_{s = - \infty}^\infty \delta\qty(r_q+\frac{x}{2} - \sqrt{\frac{2\pi}{d}} (u + d s) )]    \qty[\sum_{t = - \infty}^\infty \delta\qty(r_q-\frac{x}{2} - \sqrt{\frac{2\pi}{d}} (v + dt) )] .
\end{align}
\end{widetext}

We then use the linearity of the Wigner function and get cross terms between the computational basis states $j$ and $k$
\newpage
\begin{widetext}
\bal
W_{\ket{j}\bra{k}_\GKP}^{\CV}(r_p,r_q) &= \frac{1}{2\pi}    \int_{-\infty}^\infty \dd x e^{ir_px}\qty[\sum_{s = - \infty}^\infty \delta\qty(r_q+\frac{x}{2} - \sqrt{\frac{2\pi}{d}} (j + d s) )]    \qty[\sum_{t = - \infty}^\infty \delta\qty(r_q-\frac{x}{2} - \sqrt{\frac{2\pi}{d}} (k + dt) )]\\
&=\frac{1}{\pi} \sum_{s,t} e^{2ir_p (r_q-\sqrt{\frac{2\pi}{d}}(k+dt))}\delta\qty(2r_q-\sqrt{\frac{2\pi}{d}}[j+k+ds+dt])\\
&=\frac{1}{ 2\pi}\sum_{s,t} e^{2ir_p (r_q-\sqrt{\frac{2\pi}{d}}(k+dt))}\delta\qty(r_q-\sqrt{\frac{\pi}{2d}}[j+k+ds+dt])\\
&=\frac{1}{ 2\pi}\sum_{s,t} e^{2ir_p (r_q-\sqrt{\frac{2\pi}{d}}(k+dt-ds))}\delta\qty(r_q-\sqrt{\frac{\pi}{2d}}[j+k+dt])\\
&=\frac{1}{ 2\pi}\sum_{s,t} e^{2ir_p\sqrt{\frac{2\pi}{d}}ds}e^{2ir_p(r_q-\sqrt{\frac{2\pi}{d}}(k+dt))}\delta\qty(r_q-\sqrt{\frac{\pi}{2d}}(j+k+dt))\\
&=\frac{1}{ 2\pi}\sum_{s,t} \delta\qty(\sqrt{\frac{2d}{\pi}}r_p-s) \delta\qty(r_q-\sqrt{\frac{\pi}{2d}}(j+k+dt))e^{2ir_p(\sqrt{\frac{\pi}{2d}}(j+k+dt)-\sqrt{\frac{2\pi}{d}}(k+dt))}  \\
&= \frac{1}{2\sqrt{2\pi d}}\sum_{s,t}\delta\qty(r_p-s\sqrt{\frac{\pi}{2d}})\delta\qty(r_q-\sqrt{\frac{\pi}{2d}}(j+k+dt))e^{ir_p\sqrt{\frac{2\pi}{d}}(j-k-dt)},
\label{eq:GKPjk}
\eal
\end{widetext}
where in the second last line we used the Poisson resummation formula
\begin{align}
    \sum_{n=-\infty}^\infty e^{i2\pi n x}=\sum_{k=-\infty}^\infty \delta\qty(x-k).
\end{align}
We simplify \eqref{eq:GKP1} by using~\eqref{eq:GKPjk} and arrive at

\begin{widetext}
\begin{align}
    W_{\rho_\GKP }^{\CV}(r_p,r_q)=\frac{1}{\sqrt{8\pi d}}\sum_{u,v=0}^{d-1}\rho _{u,v}\sum_{s,t}(-1)^{s(u-v-dt)/d} \delta\qty(r_p-s\sqrt{\frac{\pi}{2d}}) \delta\qty(r_q-\sqrt{\frac{\pi}{2d}}(u+v+dt)).
     \label{eq:GKPqudit}
\end{align}
\end{widetext}

We need to find the coefficients $c_{\rho_\GKP}(l,m)$ such that
\begin{equation}\begin{aligned}
     &W_{\rho_\GKP }(r_p,r_q)\\
     &=\frac{\sqrt{d}}{\sqrt{8\pi }}\sum_{l,m} c_{\rho_\GKP}(l,m)\delta\qty(r_p-m\sqrt{\frac{\pi}{2d}})\delta\qty(r_q-l\sqrt{\frac{\pi}{2d}})
     \label{eq:atomic_app}
\end{aligned}\end{equation}
only has disjoint support for each Dirac distribution in the summation.

By inspection of Eq.~\eqref{eq:GKPqudit}, we immediately see that $\delta\qty(r_p-s\sqrt{\frac{\pi}{2d}}) $ is already in the correct form and thus will only contribute a phase with $s=m$.
Furthermore, we restrict the GKP state to one until the cell of length $\sqrt{2d\pi}$, so each $m,l$ can have $2d$ values $m,l\in \{0,1,...,2d-1\}$.
For now, let us consider $s=m=0$.
Then we get the same Dirac distribution for $q$ if $u+v+dt=l$. This requirement can be simplified if we remember that we consider only a unit cell and thus $u+v \mod d = l$.
Consequently the matrix element $c_{\rho_\GKP}(l,0)$ will be a sum of $\rho _{u,v}$ with $u+v\mod d =l$.
We can write this as
\begin{align}
    c_{\rho_\GKP}(l,0)=d^{-1}\Tr\qty(\rho  M_l)
\end{align}
with 
\begin{align}
    M_l=\sum_{\substack{u,v \in \mathds{Z}_d\\ u+v =l \; {\rm mod} \;d}} \ket{u}\bra{v}.
\end{align}
As an example, if we take qubits $d=2$
\begin{align}
    M_0&=\begin{pmatrix}
        1&0\\
        0 &1
    \end{pmatrix}\quad
    M_1=\begin{pmatrix}
        0&1\\
        1&0\\
    \end{pmatrix}
\end{align}
so we retrieve the identity and Pauli $X$.
For qutrits $d=3$ we get
\begin{align}
    M_0&=\begin{pmatrix}
        1&0 &0\\
        0 &0 &1\\
        0&1&0
    \end{pmatrix}\,
    M_1=\begin{pmatrix}
    0&1 &0\\
        1 &0 &0\\
        0&0&1
    \end{pmatrix}\,
        M_2=\begin{pmatrix}
         0&0 &1\\
        0 &1 &0\\
        1&0&0
    \end{pmatrix}
\end{align}

We now consider the general case with $m\neq 0$.
Recalling $l=u+v+dt$, the contribution for $m\neq 0$ is given by the phase factor
\bal
    \sum_m (-1)^{m(u-v-dt)/d} &= \sum_{m} (-1)^{m(2u-l)/d}\\
    &= \sum_m e^{i \pi m (2u-l)/d} = e^{-i\pi ml/d}\omega_d^{mu} 
\eal
where $\omega_d=e^{2\pi i/d}$ is the \textit{d}\,th root of unity. 
This allows us to obtain the general form of matrix elements in \eqref{eq:atomic_app} as 
\begin{align}
    c_{\rho_\GKP}(l,m)&=d^{-1}e^{-i\pi ml/d}\Tr\qty(M_l Z_d^m \rho )\\
    &=d^{-1}\Tr\qty(O_\rho(l,m) \rho )\\
    & = x_\rho(l,m).
\end{align}
This shows Proposition~\ref{prop:coefficients relation} in the case of $n=1$.


\subsection{$n$-Qudits}
In this section, we will derive the multi-qudit atomic form of GKP states.
The state of an arbitrary $n$-qudit state is given as
$\rho = \sum_{\bm{u},\bm{v} \in \mathds{Z}_d^n} \rho _{\bm{u},\bm{v}}\ket{\bm{u}}\bra{\bm{v}}$. 
The Wigner function for the GKP state that encodes this $n$-qudit state is then 
\begin{widetext}
\begin{equation}\begin{aligned}
    &W_{\rho_\GKP }^{\CV}(\bm{r}) \\
    &= \sum_{\bm{u},\bm{v} \in \mathds{Z}_d^n} \rho _{\bm{u},\bm{v}} \prod_{i=1}^n   \frac{1}{2\pi}\int_{-\infty}^\infty \dd{r_{x_i}} e^{i r_{p_i} r_{x_i}}\qty[\sum_{s_i = - \infty}^\infty \delta\qty(r_{q_i}+\frac{r_{x_i}}{2} - \sqrt{\frac{2\pi}{d}} (u_i + d s_i) )]    \qty[\sum_{t_i = - \infty}^\infty \delta\qty(r_{q_i}-\frac{r_{x_i}}{2} - \sqrt{\frac{2\pi}{d}} (v_i + dt_i) )] \\
 &= \frac{1}{(\sqrt{8\pi d})^n}  \sum_{\bm{u},\bm{v} \in \mathds{Z}_d^n} \rho _{\bm{u},\bm{v}} \prod_{i=1}^n \qty[ \sum_{s_i,t_i}(-1)^{\frac{s_i}{d}(u_i-v_i-dt_i)}    \delta\qty(r_{p_i}-\sqrt{\frac{\pi}{2d}}  s_i) \delta\qty(r_{q_i}-\sqrt{\frac{\pi}{2d}}(dt_i+u_i+v_i)) ].
 \label{eq:Wigner n qudit original}
  \end{aligned}\end{equation}
\end{widetext}

We are now ready to show Proposition~\ref{prop:coefficients relation} by confirming that 
\bal
W_{\rho_\GKP }^{\CV}(\bm{r})&=\frac{\sqrt{d}^n}{\sqrt{8\pi}^n}\sum_{\bm{l},\bm{m}} c_{\rho_\GKP}(\bm{l},\bm{m})\\
&\qquad\times \delta\qty(\bm{r_p}-\bm{m}\sqrt{\frac{\pi}{2d}})\delta\qty(\bm{r_q}-\bm{l}\sqrt{\frac{\pi}{2d}})
\label{eq:atom Wigner n qubits}
\eal
coincides with \eqref{eq:Wigner n qudit original} by taking $c_{\rho_\GKP}(\bm{l},\bm{m})=x_\rho(\bm{l},\bm{m})$.
We note that 
\begin{equation}\begin{aligned}
    &x_\rho(\bm{l},\bm{m})\\
    &=d^{-n}e^{-i\pi\bm{m} \cdot \bm{l}/d} \Tr\qty(M_{l_1}\otimes...\otimes  M_{l_n} Z_d^{m_1}\otimes ... \otimes Z_d^{m_n} \rho  )\\
    &= d^{-n}e^{-i\pi\bm{m} \cdot \bm{l}/d} \Tr\qty( M_{\bm{l}} Z_d^{\bm{m}} \rho   )
    \label{eq:matrix elements n qubits}
\end{aligned}\end{equation}
with
\begin{equation}\begin{aligned}
     &\Tr\qty( M_{\bm{l}} Z_d^{\bm{m}} \rho   )\\
     &\quad = \sum_{\bm{u},\bm{v} \in \mathds{Z}_d^n} \rho _{\bm{u},\bm{v}} \bra{\bm{v}} M_{l_1}\otimes ... \otimes M_{l_n} Z_d^{m_1}\otimes ...\otimes Z_d^{m_n}  \ket{\bm{u}}\\
     &\quad = \sum_{\bm{u},\bm{v} \in \mathds{Z}_d^n} \rho _{\bm{u},\bm{v}} \omega_d^{m_1u_1}...\omega_d^{m_n u_n} \bra{\bm{v}} M_{l_1}\otimes ... \otimes M_{l_n}  \ket{\bm{u}}.
\end{aligned}\end{equation}

Note that $\bra{\bm{v}} M_{l_1}\otimes ... \otimes M_{l_n}  \ket{\bm{u}}=1$ when
\begin{align}
    &u_i+v_i+dt_i=l_i
\end{align}
and $\bra{\bm{v}} M_{l_1}\otimes ... \otimes M_{l_n}  \ket{\bm{u}}=0$ otherwise. 
Consequently, the Wigner function \eqref{eq:atom Wigner n qubits} with the coefficients \eqref{eq:matrix elements n qubits} becomes
\begin{widetext}
\bal
    W_{\rho_\GKP }^{\CV}(\bm{r})&=\frac{\sqrt{d}^n}{\sqrt{8\pi }^n}\sum_{\bm{l},\bm{m}} x_\rho(\bm{l},\bm{m})\delta\qty(\bm{r_p}-\bm{m}\sqrt{\frac{\pi}{2d}})\delta\qty(\bm{r_q}-\bm{l}\sqrt{\frac{\pi}{2d}})\\
    &=\frac{1}{\sqrt{8\pi d}^n}\sum_{\bm{l},\bm{m}} \sum_{u,v \in \mathds{Z}_d^n} \rho _{uv}  e^{-i\pi\bm{m} \cdot \bm{l}/d} \omega_d^{m_1u_1}...\omega_d^{m_n u_n}  \bra{v} M_{l_1}\otimes ... \otimes M_{l_n}  \ket{u} \qty(\bm{r_p}-\bm{m}\sqrt{\frac{\pi}{2d}})\delta\qty(\bm{r_q}-\bm{l}\sqrt{\frac{\pi}{2d}})\\
    &=\frac{1}{\sqrt{8\pi d}^n} \sum_{u,v \in \mathds{Z}_d^n} \rho _{uv} \prod_{i=1}^n \sum_{m_i,t_i} \omega_d^{-m_i (u_i+v_i+dt_i)/2} \omega_d^{m_i u_i} \qty(r_{p_i}-m_i\sqrt{\frac{\pi}{2d}})\delta\qty(r_{q_i}-\sqrt{\frac{\pi}{2d}}(u_i+v_i+dt_i ))\\
    &=\frac{1}{\sqrt{8\pi d}^n} \sum_{u,v \in \mathds{Z}_d^n} \rho _{uv} \prod_{i=1}^n \sum_{m_i,t_i} (-1)^{m_i(u_i-v_i-dt_i)/d} \qty(r_{p_i}-m_i\sqrt{\frac{\pi}{2d}})\delta\qty(r_{q_i}-\sqrt{\frac{\pi}{2d}}(u_i+v_i+dt_i ))
\eal
\end{widetext}
which coincides with \eqref{eq:Wigner n qudit original}.
This completes the proof of Proposition~\ref{prop:coefficients relation}.

\section{Properties of the operator basis}
\label{ap:basis}

\subsection{Basic properties}
In this section, we will show the properties of the operator
\begin{align}
    O_{l,m}= e^{-i \pi m l/d} M_l Z_d^m .
\end{align}
As mentioned in the main text, the parameters are $l,m \in \mathbb{Z}_{2d}$. 
We will first show a property that involves all  $l,m \in \mathbb{Z}_{2d}$. It holds that
\begin{align}
    \sum_{ l,m \in \mathbb{Z}_{2d}} O_{l,m}&=  \sum_{ l,m \in \mathbb{Z}_{2d}} \sum_{x=0}^{d-1} \omega_d^{-m(\frac{l}{2}-x)} \ketbra{-x+l}{x}\\
    &=  \sum_{ l \in \mathbb{Z}_{2d}} \sum_{x=0}^{d-1} \sum_{ m \in \mathbb{Z}_{2d}}\omega_d^{-m(\frac{l}{2}-x)} \ketbra{-x+l}{x}\\
    &= \sum_{x=0}^{d-1} \sum_{ l \in \mathbb{Z}_{2d}} \delta_{l,2x} \ketbra{-x+l}{x}\\
    &=\mathds{1}.
\end{align}
Thus by summing over all $l,m \in \mathbb{Z}_{2d}$ we can resolve the identity using the operators $O_{l,m}$.

Using the operators, however, as an operators basis we do not need all  $l,m \in \mathbb{Z}_{2d}$.
It suffices to restrict to $l,m \in \mathbb{Z}_d$.
If we now have a value above $d$, we have
\begin{align}
    M_l&=M_{l+d}\\
    Z_d^m &= Z_d^{m+d}.
\end{align}
However, the phases can be different
\bal
   O_{l+d,m}&=  e^{-i \pi m (l+d)/d}M_l Z_d^m \\
   &=(-1)^{m} O_{l,m} \\
   O_{l,m+d}&= e^{-i \pi l (m+d)/d}M_l Z_d^m \\
   &=(-1)^{l} O_{l,m} \\
    O_{l+d,m+d}&= e^{-i \pi (m+d) (l+d)/d}M_l Z_d^m\\
    &= (-1)^{l+m+d} O_{l,m}.
    \label{eq:operatior basis phases}
\eal
So we get the same operators with a different sign.

It is easy to see that $M_l$ is Hermitian. $M_l$ is also an involution meaning $M_l^2=\mathds{1}$ because
\begin{align}
    M_l M_l &= \sum_{\substack{u+v =l \ {\rm mod}\,d\\
    u'+v'=l \ {\rm mod}\,d}}\ket{u}\bra{v}\ket{u'}\bra{v'}\\
    &=\sum_{\substack{u+v =l \ {\rm mod}\,d\\v+v' =l \ {\rm mod}\,d}}\ket{u}\bra{v'}\\
    &= \sum_{u=v'\ {\rm mod}\,d } \ket{u}\bra{v'}\\
    &= \sum_u \dyad{u}.
\end{align}
Therefore,
\bal
    O_{l,m} O_{l,m}^{\dagger} &= M_lZ_d^m \left(Z_d^m\right)^{\dagger}M_l^{\dagger}\\
    &=M_lM_l\\
    &=\mathds{1}.
    \label{eq:unitarity}
\eal
This confirms that $O_{l,m}$ is unitary.
The operator $O_{l,m}$ is also Hermitian
\bal
    O_{l,m}^{\dagger} &= \qty[e^{-i \pi m l/d}M_l Z_d^m     ]^{\dagger} \\
    &= \qty[ e^{-i \pi m l/d}\sum_{u+v =l\ {\rm mod}\,d}  \omega_d^{vm} \ket{u} \bra{v}  ]^{\dagger}\\
    &=  e^{i \pi m l/d}\sum_{u+v =l\ {\rm mod}\,d}  \omega_d^{-vm} \ket{v} \bra{u} \\
    &= e^{-i \pi m l/d}\sum_{u+v =l\ {\rm mod}\,d}  \omega_d^{m(l-v)} \ket{v} \bra{u} \\
    &=  e^{-i \pi m l/d}\sum_{u+v =l\ {\rm mod}\,d}  \omega_d^{um} \ket{v} \bra{u}\\
    &=O_{l,m}
\eal
and thus also an involution  
\begin{align}
    O_{l,m}O_{l,m}=\mathds{1}
\end{align}
because of \eqref{eq:unitarity}.
This implies that its eigenvalues are $\pm 1$.

The operators are orthogonal under the Hilbert-Schmidt inner product
\bal
    \Tr \qty(O_{l,m} O_{l',m'})&= e^{-i \pi m l/d} e^{-i \pi m' l'/d}\Tr\qty(M_l Z_d^m M_{l'}Z_d^{m'})\\
    &=\sum_{\substack{u+v =l\ {\rm mod}\,d\\v'+v\ {\rm mod}\,d=l'}}\omega_d^{mv} \omega_d^{m'v'}e^{-i \pi m l/d} e^{-i \pi m' l'/d}\\
    &\times \Tr \qty(  \ket{u}\bra{v'} )\\
    &=\sum_{\substack{u+v=l\ {\rm mod}\,d\\u+v=l'\ {\rm mod}\,d}}\omega_d^{mv+m'u}e^{-i \pi m l/d} e^{-i \pi m' l'/d}.
\eal
This is 0 if $l\neq l'$ because $l,l'\in[0,d-1]$.
For $l=l'$, we get
\bal
   \Tr \qty(O_{l,m} O_{l,m'}) &=\sum_{u+v=l\ {\rm mod}\,d} \omega_d^{mv}\omega_d^{m'u} e^{-i \pi m l/d} e^{-i \pi m' l/d}\\
    &=\sum_{u+v=l\ {\rm mod}\,d} e^{i\pi(m-m')(v-u)/d}\\
    &=\sum_{v\in\mathbb{Z}_d}\omega_d^{\Tilde{m}v}  e^{-i\pi\Tilde{m}l/d}\\
    &= e^{-i\pi\Tilde{m}l/d} \sum_{v\in\mathbb{Z}_d} \omega_d^{\Tilde{m}v}
\eal
where we set $\tilde m \coloneqq m-m'$.
This is $0$ if $\Tilde{m}\neq0$, i.e., $m\neq m'$ because a sum over roots of unity is 0.
On the other hand, when $l=l'$ and $m=m'$ we get
\begin{align}
     \Tr \qty(O_{l,m} O_{l,m}) = \sum_{v\in\mathbb{Z}_d} 1 = d.
\end{align}

In conclusion, we have
\begin{align}
     \Tr \qty(O_{l,m} O_{l',m'}) = \delta_{mm'}\delta_{ll'}d.
\end{align}

As we have seen earlier, $O_{l,m}$ are the standard Pauli operators for $d=2$. 
So for $d>2$ the operators $O_{l,m}$ are a generalization of the Pauli operators to arbitrary dimensions with the property of being an involution and Hermitian.
In general, $M_l$ and $O_{l,m}$ behave differently for even and odd dimensions, so we will separate the discussion.

\subsubsection{Odd dimensions}
In this section, we assume that the dimension $d$ is odd.
Then, the trace is
\begin{align}
    \Tr\qty(O_{l,m})&=\sum_{2x=l\ {\rm mod}\,d}e^{-i \pi m l/d} \omega_d^{mx}.
\end{align}
When $l$ is even, the solution for $2x=l\mod{d}$ is $x=l/2$ and gives $\Tr(O_{l,m})=1$. 
When $l$ is odd, the solution for $2x=l\mod{d}$ is $x=(d+l)/2$, which gives $\Tr(O_{l,m})=(-1)^m$.
These can concisely be written as 
\begin{align}
    \Tr\qty(O_{l,m})=(-1)^{ml}.
\end{align}

The operator $M_l$ has more structure, as seen in Eq.~\eqref{eq: ml def} which we report here for convenience
\bal
    M_l&=\sum_{u+v=l\ {\rm mod}\,d}\ket{u}\bra{v}
\eal
There are $d$ possibilities to fulfill this equation, so the matrix representation of $M_l$ in computational basis will have $d$ ones and the rest 0.
Furthermore, the diagonal entries $2u \mod d = l$ have only one solution for every $l$. Therefore, the matrix $M_l$ has one diagonal term $1$ and is zero otherwise.
Consequently $\Tr\qty(M_l)=1$.

Recall that we introduced the operators $O_{l,m}$ as the ones that connect the GKP state to the qudit state it encodes.
Remarkably, we can establish a direct connection between this and the phase space point operators in odd dimensions, which a priori may not have anything to do with the operators $O_{l,m}$.
The phase space point operators are defined as 
\bal
    A(a_1,a_2)
    &=d^{-1} \sum_{b_1,b_2=0}^{d-1}e^{-2i\frac{\pi}{d} (a_1,a_2)\Omega_d (b_1,b_2)^T } \omega_d^{b_1b_2/2}\\
    &\qquad\times \left(X_d^{\dagger}\right)^{b_1} \left(Z_d^{\dagger}\right)^{b_2}\\
    &= d^{-1} \sum_{b_1,b_2}\omega_d^{a_1 b_2} \omega_d^{b_1(\frac{1}{2}b_2-a_2)} \left(X_d^\dagger\right)^{b_1} \left(Z_d^\dagger\right)^{b_2}.
\eal
Using the expansion of $Z_d$ and $X_d$ in the computational basis 
\bal
      Z_d^b&=\sum_x \omega_d^{bx}\dyad{x}\\
      Z_d^{b,\dagger}&=\sum_x \omega_d^{-bx}\dyad{x}=Z(-b)\\
      X_d^b&=\sum_x \ket{x+b}\bra{x}\\
       X_d^{b\dagger}&=\sum_x \ket{x}\bra{x+b}=X(-b)
\eal
 we can write the phase space point operator as 
\bal
    A(a_1,a_2)&=d^{-1} \sum_x \sum_{b_1 b_2}  \omega_d^{b_2(a_1+\frac{1}{2}b_1)} \omega_d^{-a_2 b_1} X_d^{-b_1} Z_d^{-b_2}\\
   &= d^{-1} \sum_x \sum_{b_1 b_2}  \omega_d^{b_2(a_1-x+\frac{1}{2}b_1)} \omega_d^{-a_2 b_1} \ket{x-b_1}\bra{x}.
\eal
In order to further simplify, we need the discrete resummation formula
\begin{align}
    \frac{1}{d}\sum_{k=0}^{d-1} e^{2i\pi \frac{kn}{d}}=\delta_{0,n}.
\end{align}
The phase space point operators can then be simplified to
\begin{equation}\begin{aligned}
&A(a_1,a_2)\\
\quad&=d^{-1} \sum_x \sum_{b_1 b_2}  \omega_d^{b_2(a_1-x+\frac{1}{2}b_1)} \omega_d^{-a_2 b_1} \ket{x-b_1}\bra{x}\\
    &= \sum_x \sum_{b_1} \delta_{0, a_1-x+\frac{1}{2}b_1}\omega_d^{-a_2 b_1} \ket{x-b_1}\bra{x}\\
    &=\sum_x \sum_{b_1} \delta_{b_1,2(x-a_1)}\omega_d^{-a_2 b_1} \ket{x-b_1}\bra{x}\\
    &= \sum_x \omega_d^{-2 a_2(x-a_1)}  \ket{x-2(x-a_1)}\bra{x}\\
    &= \sum_x  \omega_d^{2a_1 a_2}  \omega_d^{-2 a_2 x} \ket{-x+2a_1}\bra{x}.
\end{aligned}\end{equation}
Using the following substitutions 
\bal
    u&=-x+2a_1\\
    v&=x\\
    u+v&=2a_1=l
\eal
we rewrite the phase space point operators as
\bal
    A(l,a_2) &= \sum_{u+v\ {\rm mod}\,d=l}  \omega_d^{-a_2 l}  \omega_d^{-2 a_2 v} \ket{u}\bra{v}\\
     &= \sum_{u+v\ {\rm mod}\,d=l}  \omega_d^{a_2 l}  \omega_d^{-2 a_2 u} \ket{u}\bra{v}.
\eal

By comparing this equation with the definition of $O_{l,m}$, we can identify $m=-2a_2$ and we get

\begin{align}
 A(a_1,a_2) = O_{2a_1,-2a_2}.
\end{align}
This shows that the operator basis $\{O_{l,m}\}_{l,m\in\mathbb{Z}_d}$ is equivalent to the phase space point operators $\{A(a_1,a_2)\}_{a_1,a_2\in\mathbb{Z}_d}$ up to permutation and phase factors. 
Indeed, the oddness of $d$ and \eqref{eq:operatior basis phases} ensure that there is a one-to-one correspondence between $a_1,a_2\in\mathbb{Z}_d$ and $l,m\in\mathbb{Z}_d$ such that $A(a_1,a_2)=O_{2a_1,-2a_2}\propto O_{l,m}$ up to phase, as $2a_2$ and $-2a_1$ respectively takes all values in $\mathbb{Z}_{d}$ by changing $a_1,a_2\in\mathbb{Z}_d$, and $O_{l_1,m_1}$ and $O_{l_2,m_2}$ coincide up to phase if $l_1=l_2\ {\rm mod}\,d$ and $m_1=m_2\ {\rm mod}\,d$.

\subsubsection{Even dimensions}

Now we investigate the operators $O_{l,m}$ for the case of even-dimensional systems.
The trace is
\begin{align}
    \Tr\qty(O_{l,m})&=\sum_{x} e^{-i \pi ml/d} \sum_{u+v=l\ {\rm mod}\,d}\omega^{mu}\bra{x}  \ket{u}\bra{v}  \ket{x}\\
    &=\sum_{2x=l\ {\rm mod}\,d}e^{-i \pi ml/d} \omega^{mx},
\end{align}
which vanishes if $l$ is odd. 
Suppose $l$ is even and write $l=2k$. Then we get 
    \begin{align}
    \Tr\qty(O_{l,m})&=\sum_{2x=l\ {\rm mod}\,d}e^{-i \pi ml/d} \omega^{mx}\\
    &=\omega^{-mk} \qty(\omega^{mk} +\omega^{m(k+\frac{d}{2})})\\
    &=1+\omega^{m\frac{d}{2}}\\
    &=1+(-1)^m.
\end{align}
We have the same behavior for the operators $M_l$. For odd dimensions the equation $u+v \mod d =l$ has $d$ solutions. Therefore, the matrix representation in computational basis will have $d$ 1's with the rest being 0.
For odd $l$, the equation $2u \mod d = l $ has no solution, implying $\Tr\qty(M_l)=0$ for odd $l$. For even $l$, the equation
$2u\ {\rm mod}\,d= l$ has two solutions, $u=l/2$ and $u= l/2 +d/2$.
Therefore, we have  $\Tr\qty(M_l)=2$.

We can make here an interesting observation. It is known that for odd dimensions the phase space point operator at the origin $A(0,0)$ acts as the parity operation
\begin{align}
    A(0,0)\ket{x}=\ket{-x}
\end{align}
for a computational basis state $\ket{x}$.
The operators $O_{l,m}$ show the same behavior for even dimensions and thus for all dimensions
\begin{align}
    O_{0,0}\ket{x}= M_0\ket{x}=\ket{-x}.
\end{align}

\subsection{Clifford Covariance}
\label{ap:cliffordcov}
This section investigates how the operators $O_{l,m}$ transform under Clifford unitaries. We will see that they behave almost equivalently to the Heisenberg-Weyl operators.

$X_d,Z_d$ generate the $d-$dimensional Heisenberg-Weyl group and $R,P,\text{SUM}$ the $d-$dimensional Clifford unitaries.
their action of computational basis state are~\cite{farinholt2014ideal} 
\begin{align}
    X_d:&\ket{j}\rightarrow\ket{j+1}\\
    Z_d:&\ket{j}\rightarrow \omega_d^{j}\ket{j}\\
    R:&\ket{j}\rightarrow \sum_{s=0}^{d-1}\omega_d^{js}\ket{s}\\
    P:&\ket{j}\rightarrow \omega_d^{j^2/2} (\omega_D \omega_{2d}^{-1})^{-j}\ket{j}\\
    \text{SUM}:& \ket{i}\ket{j}\rightarrow\ket{i}\ket{i+j}
\end{align}

A Clifford unitary $U_C$ maps the Heisenberg-Weyl operators 
\begin{align}
    P_d(a,b)= \omega_{d}^{\frac{1}{2} a b  } X_d^a Z_d^b
\end{align}
in the following way
\begin{align}
U_CP_d(u)U_C^{\dagger}= P_d(Su)    
\end{align} 
where $S$ is a $2n \times 2n$ matrix with entries over $\mathds{Z}_D$ with $D=d$ for $d$ odd and $D=2d$ for $d$ even.
Heisenberg-Weyl operators have the following commutation relations
\begin{align}
    \qty(X_d^a Z_d^b)  \qty(X_d^{a'} Z_d^{b'})= \omega_d^{(a,b)\Omega (a',b')^T}   \qty(X_d^{a'} Z_d^{b'}) \qty(X_d^a Z_d^b).
\end{align}

We are now interested in how Clifford unitaries and Pauli operators transform
\begin{align}
    O_{l,m}=\sum_{x=0}^{d-1}e^{-i \pi ml/d}  \omega_d^{mx}\ket{-x+l}\bra{x}
\end{align}
with $l,m \in [0,2d-1]$ or equivalently $\mathds{Z}_{2d}$. 
The values the computational basis states can have are $\mod d$ and the operators $Z_d^m, M_l$ are repeating for $m,l\geq d$. The difference for $m,l\geq d$ is the phase factor $\omega_d^{-ml/2}$ that repeats after $2d$.
This phase factor is important for the action of Clifford unitaries on the operators $O_{l,m}$, while it can be essentially neglected if we only want to use them as a basis.
We expect by our construction through GKP that Clifford unitaries map $O_{\bm{l},\bm{m}}$ to another $O_{\bm{l}',\bm{m}'}$.

The QFT gate $R$ transforms the operator $O_{l,m}$ as
\begin{align}
    RO_{\bm{l},\bm{m}}R^{\dagger} &=\sum_{x=0}^{d-1} e^{-i \pi ml/d}  \omega_d^{mx} R\ket{-x+l}\bra{x}R^{\dagger}\\
    &=\sum_{s,s'} e^{-i \pi ml/d} \omega_d^{sl} \sum_x \omega_d^{x(m-s-s')} \ket{s}\bra{s'}\\
    &=\sum_{s,s'} e^{-i \pi ml/d} \omega_d^{sl}   \delta_{s,m-s'}  \ket{s}\bra{s'}\\
    &=\sum_{s'}  e^{-i \pi ml/d} \omega_d^{m l} \omega_d^{-ls'} \ket{m-s'}\bra{s'}\\
    &=\sum_x e^{-i \pi ml/d}\omega_d^{-lx}\ket{m-x}\bra{x}
\end{align}
and therefore transforms the coordinates like
\begin{align}
    m&\rightarrow -l\\
    l&\rightarrow m.
\end{align}

The Phase gate $P$ behaves differently for even and odd dimensions.
For odd dimensions, we get
\bal
    P O_{\bm{l},\bm{m}} P^{\dagger} &=\sum_x \omega_d^{-m l/2}\omega_d^{mx} \omega_d^{(-x+l)(-x+l-1)/2}\omega_d^{-x(x-1)/2}\\
    &\qquad\times\ket{-x+l}\bra{x}\\
    &=\sum_x \omega_d^{-m l/2}\omega_d^{mx} \omega_d^{(l^2-l)/2} \omega_d^{x(1-l)}\ket{-x+l}\bra{x}\\
    &=\sum_x  \omega_d^{-l(m-l+1)/2}  \omega_d^{x(m-l+1)} \ket{-x+l}\bra{x}
\eal
with the coordinates transforming like
\begin{align}
    m&\rightarrow m-l+1\\
    l&\rightarrow l.
\end{align}

For even dimensions we nearly get the same result
\begin{align}
    P O_{\bm{l},\bm{m}} P^{\dagger} &=\sum_x e^{-i \pi ml/d} \omega_d^{mx} e^{i\pi(-x+l)^2/d}e^{-i\pi x^2/d}\ket{-x+l}\bra{x}\\
    &=\sum_x e^{-i \pi ml/d}\omega_d^{mx} e^{i\pi(l^2)/d} \omega_d^{-xl}\ket{-x+l}\bra{x}\\
    &=\sum_x  e^{-i\pi l(m-l)/d}  \omega_d^{x(m-l)} \ket{-x+l}\bra{x}
\end{align}
and
\begin{align}
    m&\rightarrow m-l\\
    l&\rightarrow l.
\end{align}

The action of $Z_d$ transform $O_{l,m}$ as
\begin{align}
    Z_dO_{\bm{l},\bm{m}}Z_d^{\dagger} &=\sum_x e^{-i \pi ml/d}\omega_d^{mx} \omega_d^{-x+l}\omega_d^{-x}\ket{-x+l}\bra{x}\\
    &=\sum_x  e^{-i\pi l(m-2)/d} \omega_d^{x(m-2)}\ket{-x+l}\bra{x}
\end{align}
with the coordinates transforming as
\begin{align}
    m&\rightarrow m-2\\
    l&\rightarrow l.
\end{align}
Similarly for $X_d$ we get
\begin{align}
    X_dO_{\bm{l},\bm{m}}X_d^{\dagger}&=\sum_x e^{-i \pi ml/d} \omega_d^{m x} \ket{-x+l+1}\bra{x+1}\\
    &=\sum_x e^{-i \pi ml/d} \omega_d^{m (x-1)} \ket{-x+l+2}\bra{x}\\
        &=\sum_x   e^{-i\pi m(l+2)/d}    \omega_d^{m x} \ket{-x+l+2}\bra{x}
\end{align}
and
\begin{align}
    m&\rightarrow m\\
    l&\rightarrow l+2.
\end{align}

The missing gate for the full Clifford group is the SUM gate
\begin{widetext}

\begin{align}
    \text{SUM}O_{\bm{l},\bm{m}}\otimes O_{\bm{l}',\bm{m}'}\text{SUM}^{\dagger} &=\sum_{x,y} e^{-i \pi ml/d}  e^{-i \pi m'l'/d} \omega_d^{mx} \omega_d^{m'y}
\ket{-x+l, -y-x+l'}\bra{x,y+x}\\
&=\sum_{x,y'} e^{-i \pi ml/d}  e^{-i \pi m'l'/d} \omega_d^{mx} \omega_d^{m'(y'-x)}
\ket{-x+l, -y'+l'+l}\bra{x,y'}\\
&=\sum_{x,y'} e^{-i \pi ml/d}  e^{-i \pi m'l'/d} \omega_d^{x(m-m')} \omega_d^{m'y'}
\ket{-x+l, -y'+l'+l}\bra{x,y'}
\end{align}

and
\end{widetext}
\begin{align}
    m&\rightarrow m-m'\\
    l&\rightarrow l\\
  m'&\rightarrow m'\\
    l'&\rightarrow l'+l.
\end{align}

So we can write down the matrices that transform the coordinates under the action of Clifford unitaries $U_CO_{\bm{l},\bm{m}}U_C^{\dagger} = O_{M_{U_C} [l,l',m, m']^T}$.
The matrices have the basis $(l,m)$ or $l,l',m,m'$, without the constant shifts and are given as

\begin{align}
    P:&  \begin{pmatrix}
        1&0\\
        -1&1
    \end{pmatrix}\\
    R:&\begin{pmatrix}
        0&1\\
        -1&0
    \end{pmatrix}\\
    \text{SUM}:&
    \begin{pmatrix}
        1&0&0&0\\
        1&1&0&0\\
        0&0&1&-1\\
        0&0&0&1\\
    \end{pmatrix}
\end{align}    

These matrices are all symplectic and therefore they fulfill the relation
\begin{align}
    &M^T \Omega M=\Omega\\
    &\Omega=\begin{pmatrix}
        0&\mathds{1}\\
        -\mathds{1}&0.
    \end{pmatrix}
\end{align}

 Interestingly, Clifford unitaries act on the Heisenberg-Weyl operators in a very similar way.
The matrices $S$ in the basis $a,b$ and $a_1,a_2,b_1,b_2$ are given as~\cite{farinholt2014ideal}
\begin{align}
       P':&  \begin{pmatrix}
        1&0\\
        1&1
    \end{pmatrix}\\
    R':&\begin{pmatrix}
        0&-1\\
        1&0
    \end{pmatrix}\\
    \text{SUM}:&
    \begin{pmatrix}
        1&0&0&0\\
        1&1&0&0\\
        0&0&1&-1\\
        0&0&0&1\\
    \end{pmatrix} 
\end{align}
where these matrices are over $\mathds{Z}_D$ ($D=d$ for odd $D=2d$ for even).
It was shown that these matrices generate the symplectic group over  $\mathds{Z}_D$.
We can connect our matrices (up to constant shifts) to these matrices
\begin{align}
    R^3&=R'=\begin{pmatrix}
        0&-1\\
        1&0
    \end{pmatrix}\\
    P^{d-1}&=P'=\begin{pmatrix}
        1&0\\
        1&1
    \end{pmatrix}
\end{align}
while the SUM gate is already in the correct form. So the action of a Clifford unitary on $O_{\bm{l},\bm{m}}$ is a symplectic transformation over $\mathds{Z}_{2d}$ of $l,m$. There is, however, a constant shift for $P$ in the odd dimensional case.

\subsection{Stabilizer states (Proof of Proposition~\ref{pro:pure stabilizer state})}
\label{ap:stab}
This section shows that pure stabilizer states are flat when decomposing in the operators $O_{l,m}$ meaning all operators have the same weight modulo signs.
We can write the zero state in all dimensions as 
\bal
    \dyad{0}&=\dyad{0}+\frac{1}{d}\sum_{j}\sum_{\substack{u+v=0\ {\rm mod}\,d\\v\neq 0}}w_d^{jv}\ket{u}\bra{v}\\
    &=\frac{1}{d}\sum_{j=0}^{d-1}O_{0,j}\\
    &=\frac{1}{d}\sum_{j} M_0 Z_d^j\\
    &=\frac{1}{d}\sum_j \sum_{u+v\ {\rm mod}\,d=0}\omega_d^{jv}\ket{u}\bra{v}
\eal
where in the first line we used that
\begin{align}
    \sum_{j=0}^{d-1} \omega_d^{jv}=0
\end{align}
for $v\neq 0$. We see that all nonzero $O_{l,m}$ have the same weight. Using that the operators $O_{l,m}$ are covariant under Clifford unitaries as shown in Appendix~\ref{ap:cliffordcov}, we see that every pure stabilizer state has a flat weight.


\section{Proof of Theorem~\ref{thm:characteristic}}
\label{ap:char_fct}
We compute the characteristic function for a qudit state $\rho $ encoded in a GKP state.
We need to use the following property of the displacement operator
\begin{align}
    D(\bm{r})= \prod_{j=1}^n e^{i r_{q_j} r_{p_j}/2}e^{i r_{q_j} P_j} e^{-i r_{p_j} Q_j};
\end{align}
Then, the characteristic function of a qudit encoded in a GKP state is given as
\begin{widetext}

\bal
    \chi_{\rho_\GKP }^{\CV}(\bm{r})&=\Tr\qty[\rho  D(-\bm{r})]\\
    &=\sum_{u,v\in \mathbb{Z}_d}\rho _{u,v} \bra{u}D(-\bm{r})\ket{v}\\
    &=\sum_{u,v\in \mathbb{Z}_d}\rho _{u,v}\sum_{s,t=-\infty}^\infty\bra{\sqrt{\frac{2\pi}{d}}(u+ds)} D(-\bm{r})\ket{\sqrt{\frac{2\pi}{d}}(v+dt)}_q\\
    &=\sum_{u,v}\sum_{s,t} \rho _{u,v}e^{i r_q r_p/2} e^{ir_p\sqrt{\frac{2\pi}{d}}(v+dt)} \bra{\sqrt{\frac{2\pi}{d}}(u+ds)}\ket{\sqrt{\frac{2\pi}{d}}(v+dt)+r_q}_q\\
    &=\sum_{u,v}\sum_{s,t} \rho _{u,v}e^{i r_q r_p/2} e^{ir_p\sqrt{\frac{2\pi}{d}}(v+dt)} \delta\qty(r_q-\sqrt{\frac{2\pi}{d}}(u-v-d(t-s)))\\
    &=\sum_{u,v}\sum_{s,t} \rho _{u,v} e^{i 2\pi s(r_p \sqrt{\frac{d}{2\pi}})} e^{i\frac{r_p}{2}(r_q+2\sqrt{\frac{2\pi}{d}}(v+dt))} \delta\qty(r_q-\sqrt{\frac{2\pi}{d}}(u-v-dt))\\
    &=\sqrt{\frac{2\pi}{d}}\sum_{u,v}\sum_{s,t} \rho _{u,v} \delta\qty(r_p-\sqrt{\frac{2\pi}{d}}s)\delta\qty(r_q-\sqrt{\frac{2\pi}{d}}(u-v-dt)) e^{i\frac{r_p}{2}(r_q+2\sqrt{\frac{2\pi}{d}}(v+dt))}\\
    &=\sqrt{\frac{2\pi}{d}}\sum_{u,v}\sum_{s,t} \rho _{u,v} e^{i\frac{\pi}{d}s(u+v+dt) }\delta\qty(r_p-\sqrt{\frac{2\pi}{d}}s)\delta\qty(r_q-\sqrt{\frac{2\pi}{d}}(u-v-dt)).
\eal
    
\end{widetext}

With this expression, we aim at finding the coefficient $\gamma_{\rho_\GKP}(l,m)$ for $l,m\in\mathds{Z}_d$ such that 
\begin{equation}\begin{aligned}
&\chi_{\rho_\GKP }^{\CV}(\bm{r})\\
&= \sqrt{\frac{2\pi}{d}}\sum_{l,m=-\infty}^\infty \gamma_{\rho_\GKP}(l,m)\\
&\quad \times\delta\qty(r_p-m\sqrt{\frac{2\pi}{d}})  \delta\qty(r_q-l\sqrt{\frac{2\pi}{d}}).
\end{aligned}\end{equation}
Similarly to the case for the Wigner function, we restrict to one unit cell.
Thus, the requirement for $l$ is
\begin{align}
    u-v-dt=l\\
    u-v =l \mod d
\end{align}
and therefore
\begin{align}
    \gamma_{\rho_\GKP}(l,0)=\sum_{u-v  =l \mod d}  \bra{v}\rho \ket{u} =\Tr\qty[ X_d^{l} \rho   ].
\end{align}
By simplifying the phase factor
\begin{align}
     e^{i\frac{\pi}{d}s(v+u+dt) } = e^{i\frac{\pi}{d}s (-l+2u)}= e^{-i\pi sl/d} \omega_d^{su}
\end{align}
we can write the coefficients as 
\begin{align}
    \gamma_{\rho_\GKP}(l,m)=e^{-i\pi sl/d}\Tr\qty[\rho  X_d^{l} Z_d^{m}  ].
\end{align}
Therefore, the coefficients are given by the trace over the Pauli operators in $d$ dimensions.
This can be generalized to $n$ qudits as 
\bal
    \gamma_{\rho_\GKP}(\bm{l},\bm{m})&=e^{-i\pi \bm{l}\cdot \bm{m}/d}\Tr\qty[\rho  X^{\bm{l}} Z^{\bm{m}}]\\
    & = d^n e^{-i\pi \bm{l}\cdot \bm{m}/d} \omega_d^{ -\bm{l}\cdot \bm{m}/2} \chi^\DV_\rho(\bm{l},\bm{m})^*,
\eal
which shows \eqref{eq:characteristic function coefficients}. 

Using \eqref{eq:characteristic function coefficients} and the fact that $|\chi_\phi^\DV(\bm{l},\bm{m})|=d^{-n}$ for $(4d)^n$ elements in $\bm{l},\bm{m}\in\mathds{Z}_{2d}^n$ and zero otherwise, we can get for a pure qudit state $\phi$ that 
\bal
 \|\chi_{\phi_\GKP}^\CV\|_{p,\cell} = (4d)^{n/p}\left(\frac{2\pi}{d}\right)^{n/2}.
\eal
We then get, again by using \eqref{eq:characteristic function coefficients}, that
\bal
 \|\chi_{\rho_\GKP}^\CV\|_{p,\cell} &= \left[\sum_{\bm{l},\bm{m}\in\mathds{Z}_{2d}^n}\left\{\left(\frac{2\pi}{d}\right)^{n/2} d^n|\chi_\rho^\DV(\bm{l},\bm{m})|\right\}^p \right]^{1/p}\\
 & = \left[4^n\sum_{\bm{l},\bm{m}\in\mathds{Z}_{d}^n}\left\{\left(2\pi d\right)^{n/2}|\chi_\rho^\DV(\bm{l},\bm{m})|\right\}^p \right]^{1/p}\\
 & = 4^{n/p} (2\pi d)^{n/2} \|\chi_\rho^\DV\|_p\\
 & = d^{n(1-1/p)}\|\chi_\rho^\DV\|_p,
\eal
completing the proof.


\section{Simulation algorithm}
\label{ap:simulation}
For the convenience of the reader we will use the frame notation from the works of~\cite{pashayan2015estimating}
\begin{align}
    F(\bm{\lambda})&=\frac{O_{\bm{\lambda}}}{d^n}\\
    G(\bm{\lambda})&= O_{\bm{\lambda}}\\
    \rho  &= \sum_\lambda \Tr\qty(\rho  \frac{O_{\bm{\lambda}}}{d^n}) O_\lambda=\sum_{\bm{\lambda}}G(\bm{\lambda}) \Tr\qty(\rho  F(\bm{\lambda})).
\end{align}
Unitary evolution of a state can be rewritten in this notation as
\begin{align}
    U \rho  U^{\dagger} &=\sum_{\bm{\lambda}} U G(\bm{\lambda})U^{\dagger} \Tr\qty(\rho  F(\bm{\lambda}))\\
    &=\sum_{\bm{\lambda},\bm{\lambda'}} G(\bm{\lambda'})\Tr \qty(F(\bm{\lambda'})U G(\bm{\lambda})U^{\dagger}) \Tr\qty(\rho  F(\bm{\lambda}))
\end{align}
and the output of a measurement $\Pi$ is
\begin{align}
    \Tr\qty(\Pi U\rho  U^{\dagger})
    &= \sum_{\bm{\lambda},\bm{\lambda'}} x_{\Pi}(\bm{\lambda'}) x_{U}(\bm{\lambda'},\bm{\lambda})x_{\rho }(\bm{\lambda})
\end{align}
with
\begin{align}
    x_{\rho }(\bm{\lambda})&=\Tr\qty(\rho  F(\bm{\lambda}))\\
    x_{U}(\bm{\lambda'},\bm{\lambda})&=\Tr \qty(F(\bm{\lambda'})U G(\bm{\lambda})U^{\dagger}) \\
    x_{\Pi}(\bm{\lambda})&= \Tr\qty(\Pi G(\bm{\lambda})).
\end{align}
Out of these quantities we can define the following probability distributions
\begin{align}
    P(\bm{\lambda}|\rho )&=\frac{\abs{x_{\rho }(\bm{\lambda})}}{\norm{x_{\rho }}_1}\\
    P(\bm{\lambda'}|U,\bm{\lambda})&=\frac{\abs{x_{U}(\bm{\lambda'},\bm{\lambda})}}{\norm{x_U(\bm{\lambda})}_1}\\
    \norm{x_{U}(\bm{\lambda})}_1 &= \sum_{\bm{\lambda'}}\abs{x_{U}(\bm{\lambda'},\bm{\lambda})}\\
    \norm{x_{\rho }}_1 &= \sum_{\bm{\lambda}}\abs{x_{\rho }(\bm{\lambda})}.
\end{align}
Thus we can rewrite the  Born rule probability as 
\begin{align}
    P(\Pi|U\rho  U^{\dagger} )&=\sum_{\bm{\lambda},\bm{\lambda'}} x_{\Pi}(\bm{\lambda'}) x_{U}(\bm{\lambda'},\bm{\lambda})x_{\rho }(\bm{\lambda})\\
    &=\sum_{\bm{\lambda},\bm{\lambda'}} M_{\bm{\lambda},\bm{\lambda'}} P(\bm{\lambda'}|U,\bm{\lambda})P(\bm{\lambda}|\rho )
\end{align}
with $M_{\bm{\lambda},\bm{\lambda'}}= \text{sign}\qty(x_{\rho}(\bm{\lambda}) x_{U}(\bm{\lambda'},\bm{\lambda})) x_{\Pi}(\bm{\lambda'})   \norm{x_{U}(\bm{\lambda})}_1  \norm{x_{\rho }}_1$.
The simulation strategy is to sample $\bm{\lambda} $ from $P(\bm{\lambda}|\rho )$ and then consider a possible transition to $\bm{\lambda'}$ from $P(\bm{\lambda'}|U,\bm{\lambda})$. This can easily be generalized to a sequence of unitaries of length $T$ as well.
We then define a random variable as 
\begin{align}
    M_{\Vec{\bm{\lambda}}} &= x_{\Pi}(\bm{\lambda_T}) \text{sign}(x_{\rho }(\bm{\lambda_0})) \norm{x_{\rho }}_1 \\
    &\times \prod_{t=1}^T \text{sign}(x_{U_t}(\bm{\lambda_{t}},\bm{\lambda_{t-1}}))\norm{x_{U_t}(\bm{\lambda_{t-1}})}_1.
\end{align}

The expectation value of this random variable is
\begin{align}
    \mathds{E}(M_{\Vec{\bm{\lambda}}})&= \sum_{\Vec{\bm{\lambda}}}P(\bm{\lambda_0}|\rho )\prod_{t=1}^T P(\bm{\lambda_t}|U_t,\bm{\lambda_{t-1}})  M_{\Vec{\bm{\lambda}}}\\
    &= \sum_{\bm{\Vec{\lambda}}} x_{\Pi}(\bm{\lambda_T})\prod_{t=1}^T x_{U_t}(\bm{\lambda_t},\bm{\lambda_{t-1}}) x_{\rho }(\bm{\lambda_0})
\end{align}
which is exactly the Born probability we want to estimate. The random variable output from our sampling algorithm is an unbiased estimator for the Born probability.
The number of samples needed to achieve a given precision can be computed using the Hoeffding inequality
Given a sequence of $K$ $iid$ random variables $X_j$ bounded by $\abs{X_j}\leq b$ and expected mean $\mathds{E}(X)$, the probability that $\sum_{j=1}^K X_j/K$ deviates from the mean by more than $\epsilon$ is upper bounded by
\begin{align}
    P\qty(\abs{\mathds{E}(X)-\sum_{j=1}^K \frac{X_j}{K}}\geq \epsilon )\leq 2 \exp\qty(-\frac{K\epsilon^2}{2b^2})
\end{align}
or equivalently  we can achieve precision $\abs{\mathds{E}(X)-\sum_{j=1}^K \frac{X_j}{K}}\leq \epsilon$ with probability at least $(1-p_f)$ by setting the number of samples as 
\begin{align}
    K=\left\lceil 2b^2 \frac{1}{\epsilon^2} \ln \qty(\frac{2}{p_f}) \right\rceil.
\end{align}
We then define the aggregated $l_1$ norm as 

\begin{align}
    \mathcal{M}_\rightarrow =\norm{x_{\rho }}_1  \prod_{t=1} \max_{\bm{\lambda_t}} \norm{x_{U_t}(\bm{\lambda_t})}_1    \max_{\bm{\lambda_T}} \abs{x_{\Pi}(\bm{\lambda_T})}
\end{align}
so it is the maximum $l_1$-norm over all trajectories.
This bounds the random variable from above, 
so we need at least
\begin{align}
    K\geq 2 \mathcal{M}_\rightarrow^2 \frac{1}{\epsilon^2} \ln\qty(\frac{2}{p_f})
\end{align}
samples.

\section{Hyperpolyhedral states}
\label{ap:hyper}
In this section, we investigate the phenomena of hyperpolyhedral states. In~\cite{rall2019simulation} the authors encounter hyperoctahedral states for qubit systems. They define these states as the states which have the stabilizer norm smaller than 1 or in our formulation  $ \sum_{\bm{l},\bm{m}} \abs{x_{\bm{l},\bm{m}}}<1$.
For odd dimensional states these states are equivalent to Wigner positive states  $ \sum_{\bm{l},\bm{m}} \abs{x_{\bm{l},\bm{m}}}=1$. This set is strictly bigger than the set of stabilizer states. 

For even dimensions, the question of a Wigner function is more difficult, especially related to computability, even though one can define such a quantity~\cite{raussendorf2017contextuality,bermejo2017contextuality,raussendorf2020phase}. We define the Hyperpolyhedral states similarly to the qubit case for all even dimensions with $\sum_{\bm{l},\bm{m}} \abs{x_{\bm{l},\bm{m}}}\leq 1$.
Here we show that hyperpolyhedral states exist for all even dimensions and that they are not equivalent to stabilizer states.

The computational basis states can be expanded in the operators $O_{l,m}$ as
\begin{align}
    \dyad{0}&=\frac{1}{d}\sum_{i=0}^{d-1} O_{0,i}\\
    \dyad{1}&= X_d\dyad{0}X_d^{\dagger} =\frac{1}{d } \sum_{i=0}^{d-1} O_{2,i}\\
    ...\\
    \dyad{\frac{d}{2}}&= X_d^{\frac{d}{2}}\dyad{0}X_d^{\frac{d}{2}{\dagger}} \\
    &=\frac{1}{d } \sum_{i=0}^{d-1} O(d,i)= \frac{1}{d } \sum_{i=0}^{d-1} (-1)^i O_{0,i} \\
    ...\\
    \dyad{d-1} &=  X_d^{d-1}\dyad{0}X_d^{d-1{\dagger}}\\
    & = \frac{1}{d} \sum_{i=0}^{d-1} (-1)^i O_{d-1,i}
\end{align}

Thus, we can reorder them to pairs in the following way
\begin{align}
    \dyad{0}+\dyad{\frac{d}{2}}&= \frac{1}{d} \sum_{i=0}^{d-1} (1+(-1)^i) O_{0,i}\\
    ...\\
    \dyad{k}+\dyad{k+\frac{d}{2}} &= \frac{1}{d} \sum_{i=0}^{d-1} (1+(-1)^i) O_{k,i}
\end{align}
The maximally mixed state is $\frac{\mathds{1}}{d}= \frac{1}{d}\sum_{i=0}^{d-1} \dyad{i}$.
Therefore, if we compute 
\begin{align}
    \sum_{l,m=0}^{d-1}\abs{\Tr\qty[\frac{\mathds{1}}{d} O_{l,m}]}=\frac{1}{d},
\end{align}
we get a value that is smaller than the one for a pure stabilizer state.
Therefore, one can ``hide'' magic in a product state of a magic state and the maximally mixed state or similarly including Clifford equivalent states. Since this quantity goes directly into the simulator cost, we see that the hyperpolyhedral states are easier to simulate than pure stabilizer states.

\section{Decompositions of stabilizer states}
\label{ap:dec_stab}
In this section, we show how to obtain the decompositions of stabilizer states in the basis of $O_{\bm{l},\bm{m}}$ given their stabilizers.

Every stabilizer state with a $d^n$ dimensional stabilizer group $S$ fulfills the following eigenvalue equations~\cite{gross2005finite}
\begin{align}
    \omega_d^{\bm{v}\Omega \bm{m}^T} P_d(\bm{m}) \ket{M_S,v} = \ket{M_S,v} 
\end{align}
where $M_S$ is the space of coordinates associated with the stabilizer group and $\bm{v}$ is the coordinate of one Heisenberg-Weyl operator. Note that $\bm{m} \in \mathds{Z}_d^{2n}$ and that the phases are taken care of by the phase factor $ \omega_d^{\bm{v}\Omega \bm{m}^T}$.
This $\bm{v}$  takes care of the phase in front of the Heisenberg-Weyl operator.
The stabilizer group $S$ and in turn the set of coordinates are generated by $n$ Heisenberg-Weyl operators $S=\langle S_1,...,S_n\rangle$ or $n$ coordinates $M_S=\langle\bm{s_1},...,\bm{s_n}\rangle$ respectively. The set $M_S$ includes a linear combinations involving the $n$ generators $\bm{s_i}$ with coefficients $k_i \in \mathds{Z}_d$.
The characteristic function of a stabilizer state $\dyad{\phi}$ can be represented as~\cite{gross2006hudson}
\begin{align}
    \chi_{\dyad{\phi}}^\DV(\bm{a})=\frac{1}{d^n} \omega_d^{\bm{v}^T\Omega \bm{a}} \delta_{M_S}(\bm{a})
\end{align}
where $\delta_{M_S}(\bm{a})$  is the indicator function that $\delta_{M_S}(\bm{a})=1$, iff $\bm{a} \in M_S$ and $0$ otherwise.

For odd dimensions, the phase space point operators $A(a_1,a_2)$ had a one-to-one correspondence with the operators $O_{l,m}$. Something similar holds for even dimensions as well. In that case, the operators $O_{l,m}$ are directly connected with operators $\Tilde{A}(a_1,a_2)$ that are identically defined as the phase space point operators but do not fulfill the same set of properties.

In the proof to show the connection between the phase space point operators in odd dimensions and the operators $O_{l,m}$, we used the resummation formula
\begin{align}
    \frac{1}{d}\sum_{k=0}^{d-1}e^{2\pi i \frac{k n}{d}}= \delta_{0,n}.
\end{align}

For even dimensions, the sums appear with $2d$ instead of just $d$ and then
\begin{align}
    \frac{1}{d}\sum_{k=0}^{d-1}e^{2\pi i \frac{k n}{2d}} \neq\delta_{0,n},
\end{align}
so one cannot easily use the discrete resummation formula.
However, the equation can be modified to hold in all dimensions. In even dimensions, it holds that 
\begin{align}
    \frac{1}{2d}\sum_{k=0}^{2d-1}e^{2\pi i \frac{k n}{2d}}= \frac{1}{\Tilde{d}} \sum_{k=0}^{\Tilde{d}-1} e^{2\pi i \frac{kn}{\Tilde{d}}} = \delta_{0,n}.
\end{align}
We see that by doubling the domain of the sum, we recover the discrete resummation formula.
So we can use the discrete resummation formula in all dimensions by considering
\begin{align}
    \frac{1}{D}\sum_{k=0}^{D-1}e^{2\pi i \frac{k n}{D}}= \delta_{0,n}
\end{align}
where $D$ is the quantity defined in \eqref{eq:D definition} and in consequence
\begin{align}
    \frac{1}{D^n} \sum_{\bm{x}\in \mathds{Z}_D^{2n}}       \omega_D^{-(\bm{u}+\bm{v})\Omega \bm{x}^T}= \delta_{0,\bm{u}+\bm{v}}.
    \label{eq:genresum}
\end{align}
We define the symplectic Fourier transform as 
\begin{align}
    (\mathcal{F} f)(\bm{a}) =\frac{1}{D^n}\sum_{\bm{b}\in \mathds{Z}^{2n}_D } \omega_D^{-\bm{a}^T\Omega \bm{b}} f(\bm{b}).
\end{align}
We see here again that the parameters $\bm{a}$ are not over $ \mathds{Z}^{2n}_d$ but over $ \mathds{Z}^{2n}_D$ as mentioned before.
In order to differentiate between the phase-space point operators in odd dimensions $A(\bm{a}) = \frac{1}{d^n}\sum_{\bm{b}\in \mathds{Z}^{2n}_d } \omega_d^{-\bm{a}\Omega \bm{b}^T} P_d^\dagger (\bm{b})$ with their intimate relation with the discrete Wigner function, we define the equivalently defined operator $\Tilde{A}(\bm{{a}})= \frac{1}{D^n}\sum_{\bm{b}\in \mathds{Z}^{2n}_D } \omega_D^{-\bm{a}\Omega \bm{b}^T} P^\dagger (\bm{b})$ for even dimensions. 

Then consequently for even dimension, we can rewrite $\Tilde{A}(\bm{{a}})$ as
\begin{align}
    \Tilde{A}(\bm{a}) &= \frac{1}{{D}}\sum_{\bm{b}\in \mathds{Z}^{2n}_D } \omega_D^{-\bm{a}^T\Omega \bm{b}} P_d^\dagger (\bm{b})\label{eq:atilde}\\
    &= \frac{1}{{D}}\sum_{\bm{b}\in \mathds{Z}^{2n}_D, \bm{x}\in \mathds{Z}^{n}_d } \omega_d^{\bm{b_2} (\frac{\bm{a_1}}{2}+\frac{\bm{b_1}}{2}-\bm{x})} \omega_D^{-\bm{a_2}\bm{b_1}} \ket{\bm{x}-\bm{b_1}}\bra{\bm{x}}\\
    &=\frac{1}{D}\sum_{\bm{b_1} \in \mathds{Z}^{2n}_D, \bm{x}\in \mathds{Z}^{n}_d}\delta_{\bm{b_1}, 2\bm{x}-\bm{a_1}} \omega_D^{-\bm{a_2b_1}} \ket{\bm{x}-\bm{b_1}}\bra{\bm{x}}\\
    &=\frac{1}{D} \sum_{ \bm{x}\in \mathds{Z}^{n}_d} \omega_d^{\bm{a_1a_2}/2} \omega_d^{-\bm{a_2 x}} \ket{\bm{a_1}-\bm{x}}\bra{\bm{x}} 
    \label{eq:phase point and operator basis}
\end{align}
which is equivalent to $O_{\bm{l},\bm{m}}$ for $\bm{l}=\bm{a_1}, \bm{m}=-\bm{a_2}$.

We transform the characteristic function to get the coefficients $x_{\dyad{\phi}}$  corresponding to the operators $O_{l,m}$ for a stabilizer state $\phi$ as
\begin{align}
    \mathcal{F}(\chi_{\dyad{\phi}}^\DV)(\bm{a}) &= \frac{1}{(2d)^n} \frac{1}{d^n} \sum_{\bm{b}\in \mathds{Z}^{2n}_{2d} } \omega_{2d}^{-\bm{a}^T\Omega \bm{b}} \delta_{M_s}(\bm{b}) \omega_d^{\bm{v}\Omega \bm{b}^T}\\
    &=\frac{1}{(2d)^n} \frac{1}{d^n} \sum_{\bm{b}\in M_S\cup M_S+d}\omega_{2d}^{-(\bm{a}-\bm{2\bm{v}})^T\Omega\bm{b}}.
\end{align}
We used that $M_S$ was defined on $\mathds{Z}^{2n}_d$, but the sum goes over $ \mathds{Z}^{2n}_{2d}$ so we need to take this into account when dealing with the phase factors.
We write $ M_S\cup M_S+d$ as the extension from $\mathds{Z}^{2n}_d$ to $\mathds{Z}^{2n}_{2d}$. This set is generated by the same generators $\bm{s_i}$ but includes now lienar combinations with coefficients $k_i \in \mathds{Z}_{2d}$.
We can simplify the sum by using the generators of the coordinate space $M_S=\langle \bm{s_1},...,\bm{s_n}\rangle$ to
\begin{align}
      \mathcal{F}(\chi_{\dyad{\phi}}^\DV)(\bm{a}) &=\frac{1}{(2d)^n} \frac{1}{d^n} \sum_{\bm{b}\in M_S\cup M_S+d}\omega_{2d}^{-(\bm{a}-\bm{2\bm{v}})^T\Omega\bm{b}}\\
      &=\frac{1}{(2d)^n} \frac{1}{d^n}\prod_{i=1}^n \qty(\sum_{k_i=0}^{2d-1}\omega_{2d}^{-(\bm{a}-\bm{2\bm{v}})^T\Omega \qty[k_i\bm{s_i}]}   )\\
      &=\frac{1}{(2d)^n} \frac{1}{d^n}\prod_{i=1}^n \qty(\sum_{k_i=0}^{2d-1}\omega_{2d}^{-k_i(\bm{a}-\bm{2\bm{v}})^T\Omega \bm{s_i}}   )\\
        &=\frac{1}{d^n} \delta_{M_S+2\bm{v}}(\bm{a}).
    \label{eq:coefficientsPaulicond}
\end{align}
From the first to the second line we decomposed the elements $\bm{b}\in  M_S\cup M_S+d $ using the generators $M_S\cup M_S+d$. Each element $\bm{b}$ can be decomposed into a linear combination of the generators $\bm{s_i}$ and coefficients $k_i \in \mathds{Z}_{2d}$.
In the last line, we used the resummation formula~\eqref{eq:genresum} and saw that the sum is $0$ by using except in the case where $(\bm{a}-2\bm{v})\Omega \bm{b}^T =0$ and thus the Pauli operators in the stabilizer group commute with Pauli operators with coordinates $\bm{a}-2\bm{v}$. This implies that $\bm{a}-2\bm{v} \in M_S$ since $S$ is a stabilizer group with the maximal number of commuting Pauli operators. 
As shown in \eqref{eq:phase point and operator basis}, the expression in Eq.~\eqref{eq:coefficientsPaulicond} coincides with $x_{\dyad{\phi}}(l=\bm{a}_1,m=-\bm{a_2})$.
A few comments are in order. We have shown that we can write every pure stabilizer state using $d^n$ operators $O_{\bm{l},\bm{m}}$ that all have the phase $+1$
\begin{align}
    \dyad{\phi}&=\sum_{(\bm{l},\bm{m}) \in \mathds{Z}^{2n}_{2d}} x_{\dyad{\phi}} (\bm{l},\bm{m}) O_{\bm{l},\bm{m}}\\
    &= \frac{1}{d^n}\sum_{(\bm{l},-\bm{m}) \in M_s+2\bm{v}}  O_{\bm{l},\bm{m}}.
\end{align}
Note that the sums go over $D$ and not $d$, which makes a difference in even dimensions.
As we know, the operators $O_{\bm{l},\bm{m}}$ repeat with period $d$ with the opposite sign. Let us take the example of qubits $Z_2= O_{0,1}$ while $-Z_2= O_{2,1}$. So if we constrain $(\bm{l},\bm{m}) \in \mathds{Z}^{2n}_d$ we can get phases $\pm 1$, while if we allow for all $(\bm{l},\bm{m}) \in \mathds{Z}^{2n}_D$ we get decompositions with only $+1$ signs.


\section{Characteristic function: Faithfulness}
\label{ap:faith}
In this section, we will prove the faithfulness property  for the 1-norm of the characteristic function  $\norm{\chi_{\dyad{\phi}}^\DV}_1$. The characteristic function of a states returns the coefficients of that state expanded in Pauli basis.
A pure stabilizer state is by definition stabilized by $d^n$  commuting Pauli operators and thus can be represented with $d^n$ equally weighted Pauli operators.
Thus if $\dyad{\psi}$ is a stabilizer state then $\norm{\chi_{\dyad{\phi}}^\DV}_1=1$.

We need to show the other direction as well namely that if  $\norm{\chi_{\dyad{\phi}}^\DV}_1=1$ then $\dyad{\phi}$ is a stabilizer state.
We will use the purity of the states we are considering.
The pure state has to have purity 1 and can be written using the characteristic function as
\begin{align}
    \Tr\qty[\dyad{\phi}^2]=\frac{1}{d^{n}} \sum_{\bm{a}\in \mathds{Z}_d^{2n}} \abs{\chi^\DV_{\dyad{\phi}}(\bm{a})}^2=1.
\end{align}
The characteristic function is bounded by $\abs{\chi^\DV_{\dyad{\phi}}(\bm{a})}\leq 1$.
The magic meausure we are considering is the 1-norm of the discrete characteristic function
\begin{align}
    \norm{\chi_{\dyad{\phi}}^\DV}_1= \frac{1}{d^{n}} \sum_{\bm{a}\in \mathds{Z}_d^{2n}} \abs{\chi^\DV_{\dyad{\phi}}(\bm{a})}.
\end{align}
As we have seen, $\abs{\chi_{\dyad{\phi}}^\DV}$ is 1 for $d^n$ values corresponding to the stabilizer group and is 0 otherwise. By observing the purity condition we see that  $\abs{\chi_{\dyad{\phi}}^\DV}$ needs to be non-zero for at least $d^n$ values. 
Assume now a state $\ket{\psi}$ has a decomposition including $d^n$ Pauli operators with coefficients $\pm \frac{1}{d^n}$. Such a state satisfies $\bra{\psi}P_d\ket{\psi}=1$ for every $P_d$ in the decomposition. $\ket{\phi}=P_d\ket{\psi}$ is also a proper pure quantum state so it has to hold that $\ket{\phi}=\ket{\psi}$, which immediately implies $P_d\ket{\psi}=\ket{\psi}$. This means $\ket{\psi}$ is a simultaneous $+1$ eigenstate of all $d^n$ Pauli operators, which implies that the Pauli operators are commuting. This is precisely the definition of a stabilizer state.
Thus if the characteristic function is non-zero for exactly $d^n$ values, it is a stabilizer state. Now let us assume a state that is non-zero for more than $d^n$ values. This set of values we call $R$.
The purity condition is then
\begin{align}
    \Tr\qty[\dyad{\phi}^2]=\frac{1}{d^{n}} \sum_{\bm{a}\in R} \abs{\chi^\DV_{\dyad{\phi}}(\bm{a})}^2=1
\end{align}
with $\text{card}(R)>d^n$. Since  $\text{card}(R)>d^n$ it has to hold that at $\abs{\chi^\DV_{\dyad{\phi}}(\bm{a})}<1$ for at least two $\bm{a_1},\bm{a_2}$. Assume for simplicity now that $\text{card}(R)=d^n+1$ and two $\abs{\chi^\DV_{\dyad{\phi}}(\bm{a})}<1$. 
Then it has to hold that
\bal
    \Tr\qty[\dyad{\phi}^2]&=\frac{1}{d^{n}} \sum_{\bm{a}\in R} \abs{\chi^\DV_{\dyad{\phi}}(\bm{a})}^2\\
    &= \frac{1}{d^n} \qty( d^n-1 + \abs{\chi^\DV_{\dyad{\phi}}(\bm{a_1})}^2 +\abs{\chi^\DV_{\dyad{\phi}}(\bm{a_2})}^2)\\
    &=1
\eal
and in consequence
\begin{align}
    \abs{\chi^\DV_{\dyad{\phi}}(\bm{a_1})}^2 +\abs{\chi^\DV_{\dyad{\phi}}(\bm{a_2})}^2 =1.
\end{align}
As long as $\abs{\chi^\DV_{\dyad{\phi}}(\bm{a_1})}\neq 0$ and $\abs{\chi^\DV_{\dyad{\phi}}(\bm{a_1})}\neq 0$, which would violate our assumption of 
 $\text{card}(R)>d^n$, it holds that
 \begin{align}
       \abs{\chi^\DV_{\dyad{\phi}}(\bm{a_1})} +\abs{\chi^\DV_{\dyad{\phi}}(\bm{a_2})} >1,
 \end{align}
 since $x>x^2$ for every positive real number $x<1$.
 In consequence, the 1-norm is then
 \begin{align}
    \norm{\chi_{\dyad{\phi}}^\DV}_1&= \frac{1}{d^{n}} \sum_{\bm{a}\in R} \abs{\chi^\DV_{\dyad{\phi}}(\bm{a})} \\
    &= \frac{1}{d^n}\qty(d^n+ \abs{\chi^\DV_{\dyad{\phi}}(\bm{a_1})} +\abs{\chi^\DV_{\dyad{\phi}}(\bm{a_2})})>1.
\end{align}
The same argument applies to any increase in cardinality above $d^n$ but is easiest to see for two coefficients $<1$.
This immediately gives us faithfulness for the discrete characteristic function if restricted to pure states. $\norm{\chi_{\dyad{\phi}}^\DV}_1=1$ if and only if $\dyad{\phi}$ is a pure stabilizer state.


\section{Simulating GKP}
\label{ap:Simulating GKP}

In this section, we provide additional details on the simulation algorithm for circuits that use ideal GKP codewords, Gaussian unitary operations and Gaussian measurements, i.e. homodyne detection, that we introduced in the main text.

We want to simulate the run of a quantum circuit, where we in the end obtain a sample $\bm{x}$ of a homodyne measurement.
The sample $\bm{x}$ is drawn from the probability distribution of obtaining a certain measurement outcome $\bm{x}$
\begin{align}
    P(\bm{x})=\Tr[ \Pi_{\bm{x}} U_G \rho U_G^\dagger ]
\end{align}
where $\Pi_{\bm{x}}$ describes the measurement effect of obtaining measurement outcome $\bm{x}$.
However, a GKP states are not quantum states, since they are non-normalizable. 
To differentiate between a probability $P(x)$ and the quantity we obtain by measuring an ideal GKP state, we call the latter $\Tilde{P}(x)$.

We will assume now an ideal GKP codeword encoding a multi-qudit state $\rho$. We then rewrite the ``probability'' we want to sample from as
\begin{align}
     \Tilde{P}(\bm{x})&= \int \dd \bm{r}\; W^{\CV}_{\Pi_{\bm{x}}}(\bm{r}) W^{\CV}_{U_G\rho_\GKP U_G^\dagger}(\bm{r})\\
     &=\int \dd \bm{r}\; W^{\CV}_{\Pi_{\bm{x}}}(\bm{r}) W^{\CV}_{\rho_\GKP}(S^{-1}\bm{r}+\bm{d})\\
     &= \int \dd \bm{p}\; W^{\CV}_{\rho_\GKP}(S^{-1} [\bm{p},\bm{x}]^T+\bm{d}).
\end{align}

We then insert the atomic form defined in Eq.~(\ref{eq:atomic})
\begin{align}
     \Tilde{P}(\bm{x}) &= \frac{\sqrt{d}^n}{\sqrt{8\pi}^n} \sum_{\bm{u}=-\infty}^\infty x_{\rho}(\bm{u})  \int \dd \bm{p}\;\delta(S^{-1} [\bm{p},\bm{x}]^T+\bm{d}-\sqrt{\frac{\pi}{2d}}\bm{u})\\
     &= \frac{\sqrt{d}^n}{\sqrt{8\pi}^n} \sum_{\bm{u}=-\infty}^\infty x_{\rho}(\bm{u})  \int \dd \bm{p}\;\delta( [\bm{p},\bm{x}]^T+ S \bm{d}-\sqrt{\frac{\pi}{2d}} S \bm{u}).
\end{align}
We used the following property of Dirac distributions.
It holds for Dirac distributions composed with a function $g: \mathds{R}^{2n} \rightarrow \mathds{R}^{2n}$ that
\begin{align}
    \delta\qty(g([\bm{p},\bm{x}]^T))= \sum_i \frac{\delta \qty([\bm{p},\bm{x}]^T -[\bm{x_{i,0}},\bm{p_{i,0}}]^T )}{\abs{\det[g'([\bm{x_{i,0}},\bm{p_{i,0}}]^T)]}}
\end{align}
where $g'$ is the first derivative and $[\bm{x_{i,0}},\bm{p_{i,0}}]^T$ are the roots of $g$. In our case $g$ is linear and thus $\abs{\det[g'([\bm{x_{i,0}},\bm{p_{i,0}}]^T)]}=\abs{\det[S^{-1}]}=1$, since all symplectic matrices have $\det[S]=1$.

\begin{widetext}

The coefficients $ x_{\rho}(\bm{u})$ repeat with a period of $2d$. So we write
\begin{align}
     \Tilde{P}(\bm{x})  &= \frac{\sqrt{d}^n}{\sqrt{8\pi}^n} \sum_{\bm{u}=-\infty}^\infty x_{\rho}(\bm{u})  \int \dd \bm{p}\;\delta\qty( [\bm{p},\bm{x}]^T+ S \bm{d}-\sqrt{\frac{\pi}{2d}} S \bm{u})\\
      &=\frac{\sqrt{d}^n}{\sqrt{8\pi}^n} \sum_{\bm{u}\in \mathds{Z}_{2d}^{2n}} x_{\rho}(\bm{u}) \sum_{\bm{n}\in \mathds{Z}^{2n}} \int \dd \bm{p}\;\delta\qty( [\bm{p},\bm{x}]^T+ S \bm{d}-\sqrt{\frac{\pi}{2d}} S \bm{u}-\sqrt{2\pi d}S\bm{n}).
\end{align}
\end{widetext}

Let us investigate the last integral
\begin{align}
    \sum_{\bm{n}\in \mathds{Z}^{2n}}\int \dd \bm{p}\;\delta\qty( [\bm{p},\bm{x}]^T+ S \bm{d}-\sqrt{\frac{\pi}{2d}} S \bm{u}-\sqrt{2\pi d}S\bm{n})
\end{align}
This is only $\neq 0$ iff
\begin{align}
    [\bm{p},\bm{x}]^T = \sqrt{\frac{\pi}{2d}} S \bm{u}+\sqrt{2\pi d}S\bm{n}-S\bm{d}.
\end{align}
Integration over $\bm{p}$ removes the Dirac distributions over $\bm{p}$
\begin{align}
    \sum_{\bm{n}\in \mathds{Z}^{2n}}\delta\qty( \bm{x}+ \Tr_{\bm{p}}\qty[S \bm{d}-\sqrt{\frac{\pi}{2d}} S \bm{u}-\sqrt{2\pi d}S\bm{n}])
\end{align}
To conclude
\begin{widetext}
  \begin{align}
     \Tilde{P}(\bm{x})=\frac{\sqrt{d}^n}{\sqrt{8\pi}^n} \sum_{\bm{u}\in \mathds{Z}_{2d}^{2n}} x_{\rho}(\bm{u})    \sum_{\bm{n}\in \mathds{Z}^{2n}}\delta\qty( \bm{x}+ \Tr_{\bm{p}}\qty[S \bm{d}-\sqrt{\frac{\pi}{2d}} S \bm{u}-\sqrt{2\pi d}S\bm{n}]).
\end{align}  
\end{widetext}
To ge a sample $\bm{x}$ drawn for the distribution we first sample a $\bm{u}$ according to $\frac{\abs{x_{\rho}}}{\norm{x_{\rho}}_1}$ which leads to the expression
\begin{widetext}
  \begin{align}
     \Tilde{P}(\bm{x})=\frac{\sqrt{d}^n}{\sqrt{8\pi}^n} \sum_{\bm{u}\in \mathds{Z}_{2d}^{2n}} \text{sign}( x_{\rho}(\bm{u})  )\norm{x_{\rho}}_1 \frac{\abs{x_{\rho}}(\bm{u})}{\norm{x_{\rho}}_1}    \sum_{\bm{n}\in \mathds{Z}^{2n}}\delta\qty( \bm{x}+ \Tr_{\bm{p}}\qty[S \bm{d}-\sqrt{\frac{\pi}{2d}} S \bm{u}-\sqrt{2\pi d}S\bm{n}]).
\end{align}  
\end{widetext}
Then we sample uniformly a $\bm{n} \in \mathds{Z}_{2d}^{2n}$ and return $\bm{x} = \Tr_{\bm{p}}\qty[-S \bm{d}+\sqrt{\frac{\pi}{2d}} S \bm{u}+\sqrt{2\pi d}S\bm{n}]$ as the index of the measurement result.

\bibliographystyle{apsrmp4-2}
\bibliography{myref}

\end{document}